\documentclass{article}
\usepackage{arxiv}

\usepackage[utf8]{inputenc} % allow utf-8 input
\usepackage[T1]{fontenc}    % use 8-bit T1 fonts
\usepackage{hyperref}       % hyperlinks
\usepackage{url}            % simple URL typesetting
\usepackage{booktabs}       % professional-quality tables
\usepackage{amsfonts}       % blackboard math symbols
\usepackage{nicefrac}       % compact symbols for 1/2, etc.
\usepackage{microtype}      % microtypography
\usepackage{lipsum}

\usepackage{comment}
\usepackage{makecell}
\usepackage{wrapfig}
\usepackage{stfloats}

\usepackage{amsthm,amsmath,amssymb}
\usepackage{graphicx}
\usepackage{subcaption}
\usepackage{listings}
\usepackage{courier}
\usepackage[dvipsnames]{xcolor}
\usepackage[justification=centering,skip=7pt,labelfont=bf]{caption}

\renewcommand{\paragraph}[1]{\vspace{0.5em}\noindent\textbf{#1}.}

\usepackage{float}
\floatstyle{plaintop}
\restylefloat{table}

\newboolean{showcomments}
\setboolean{showcomments}{true}
\ifthenelse{\boolean{showcomments}}
{ \newcommand{\mynote}[3]{
    \fbox{\bfseries\sffamily\scriptsize#1}
{\small$\blacktriangleright$\textsf{\emph{\color{#3}{#2}}}$\blacktriangleleft$}}
}{
\newcommand{\mynote}[3]{}}

\usepackage{listings}
\usepackage{etoolbox}
\newtoggle{InString}{}% Keep track of if we are within a string
\togglefalse{InString}% Assume not initally in string
\newcommand*{\ColorIfNotInString}[1]{\iftoggle{InString}{#1}{\color{blue}#1}}%
\newcommand*{\ProcessQuote}[1]{#1\iftoggle{InString}{\global\togglefalse{InString}}{\global\toggletrue{InString}}}%
\newcommand{\indentrule}{\color{code_indent}\vrule\hspace{2pt}}
\definecolor{code_indent}{HTML}{CCCCCC}

\newenvironment{figureAsListing}
    {
    \addtocounter{figure}{-1}
    \refstepcounter{lstlisting}
     
    \begin{figure}[!htbp]
        
        \centering
    }
    { 
        \end{figure} 
    }

\usepackage{multicol}

\newcommand{\ett}{\emph{ETT}}
\newcommand{\findRoot}{\texttt{find\_root(u)}}
\newcommand{\removeEdge}{\texttt{remove\_edge(u, v)}}
\newcommand{\addEdge}{\texttt{add\_edge(u, v)}}
\newcommand{\connected}{\texttt{connected(u, v)}}

\newtheorem{theorem}{Theorem}[section]

\newtheorem{lemma}[theorem]{Lemma}

\def\O{\mathcal{O}}

\title{A Scalable Concurrent Algorithm \\ for Dynamic Connectivity}

\author{
  Alexander Fedorov \\
  JetBrains Research and Higher School of Economics\\
  \texttt{alexander.fedorov@jetbrains.com} \\
  \And
  Nikita Koval \\
  JetBrains Research\\
  \texttt{nikita.koval@jetbrains.com}
  \And
  Dan Alistarh\\
  IST Austria\\
  \texttt{dan.alistarh@ist.ac.at} \\
}

\begin{document}

\maketitle

\begin{abstract}
\emph{Dynamic Connectivity} is a fundamental algorithmic graph problem, motivated by a wide range of applications to social and communication networks
and used as a building block in various other algorithms, such as the bi-connectivity and the dynamic minimal spanning tree problems.
In brief, we wish to maintain the connected components of graph under dynamic edge insertions and deletions. 
In the sequential case, the problem has been well-studied from both theoretical and practical perspectives. 
However, much less is known about efficient \emph{concurrent} solutions to this problem. This is the gap we address in this paper.

We start from one of the classic data structures used to solve this problem, the Euler Tour Tree. 
Our first contribution is a non-blocking \emph{single-writer} implementation of it. 
We leverage this data structure to obtain the first truly concurrent generalization of dynamic connectivity, which preserves the time complexity of its sequential counterpart, but is also scalable in practice. 
To achieve this, we rely on three main techniques.
The first is to ensure that \emph{connectivity queries}, which usually dominate real-world workloads, are non-blocking. 
The second non-trivial technique expands the above idea by making all queries that do not change the connectivity structure \emph{non-blocking}. 
The third ingredient is applying fine-grained locking for updating the connected components, which allows operations on disjoint components to occur in parallel.

We evaluate the resulting algorithm on various workloads, executing on both real and synthetic graphs. The results show the efficiency of each of the proposed optimizations; the most efficient variant improves the performance of a coarse-grained based implementation on realistic scenarios up to 6x on average and up to 30x when connectivity queries dominate.
% First, we introduce a concurrent lock-free single-writer Euler tour tree algorithm. Then, we use it to propose a concurrent generalization of the classical sequential algorithm, which preserves theoretical efficiency and is scalable at the same time. Three key optimizations help to achieve this. The first one makes connectivity queries, which usually dominate in real-world workloads, non-blocking. The second optimization is fine-grained locking for components of connectivity; that makes it possible to parallelize work on different components straightforwardly.  The last one expands the idea of non-blocking operations by making most of the queries (the ones that do not change internal components representation) non-blocking. 
\end{abstract}

\section{Introduction}

\emph{Dynamic connectivity} is a fundamental algorithmic problem on graphs under edge updates, motivated by a wide range of applications in the context of social and communication networks. It is also used as a key building block for various other algorithms, such as the bi-connectivity and the dynamic minimal spanning tree problems. Essentially, the goal is to maintain a query on whether two vertices are in the same connected component under edge insertions and deletions. Thus, the following operations on undirected graph $G$ should be supported:
\begin{itemize}
    \item \texttt{addEdge(u, v)} adds a new edge $(u, v)$ to $G$;
    \item \texttt{removeEdge(u, v)} removes the edge ${(u, v)}$ from $G$;
    \item \texttt{connected(u, v): Bool} checks whether $u$ and $v$ are in the same component of connectivity, returning \texttt{true} in this case.
\end{itemize}
We refer to both \texttt{addEdge(u, v)} and \texttt{removeEdge(u, v)} as \emph{write} operations or \emph{modifications}, and to \texttt{connected(u, v)} as a \emph{read} operation or \emph{query}.
Intuitively, read operations are easier to parallelize, while modifications require non-trivial synchronization.

One classic approach to solve the dynamic connectivity problem is to use \emph{Euler Tour Trees (ETT)}~\cite{Henzinger1999} to maintain a spanning forest of the graph. This provides a relatively simple way to implement both edge additions and connectivity queries in $\O(\log N)$ time per operation, where $N$ is the number of graph vertices. However, edge removals are non-trivial. The challenge is that, when a spanning edge is removed, there can be an edge that reconnects the spanning trees; this edge needs to be identified and added to the spanning forest. The classic algorithms proposed by Henzinger et al.~\cite{Henzinger1999} and Holm et al.~\cite{Holm2001} partition non-spanning edges into levels, to achieve amortized polylogarithmic time complexity per modification. In turn, these algorithms were improved to $\O(\log N (\log \log N)^2)$ and $\O(\frac {\log^2 N} {\log \log N})$ time complexity in~~\cite{Wulff-Nilsen2013}~and~\cite{Huang2017}, respectively, using similar ideas. Interestingly, these algorithms are close to the theoretical lower bounds~\cite{Patrascu2005,Patrascu2011}.

Even though this problem is well-studied in the sequential world in both theoretical and practical aspects~\cite{seq-experiments}, to our knowledge there is still no truly concurrent algorithm with polylogarithmic time per operation. Nonetheless, we highlight a few approaches which take advantage of parallelism for variants of this problem. 

Aksenov et al. introduced a \emph{parallel combining} technique~\cite{Aksenov2018}. which allows performing read operations in parallel. In contrast, all modifications are performed sequentially and cannot be executed in parallel with reads. To evaluate this technique, they applied it to the dynamic connectivity problem so that concurrent \texttt{connected(u, v)} invocations can be executed in parallel {---} the idea is similar to the readers-writer lock. However, this is essentially a sequential algorithm which leverages the \emph{flat combining} technique~\cite{flat-combining}. 

Acar et al.~\cite{Acar2019} proposed an algorithm which handles operations \emph{in batches}, allowing some operations in a batch to be processed in parallel and therefore more efficiently. One shortcoming of this approach is that it requires requests to be batched, and, as such, it is not a truly concurrent solution. 

% \paragraph{Our Contribution.}

%Similarly to the modern sequential algorithms,

%The standard lock-free implementation tries to make a progress optimistically and re-starts the operation.

\paragraph{Our Contribution}
In this paper, we present the first concurrent solution for the dynamic connectivity problem. Our algorithm is \emph{linearizable}~\cite{Herlihy1990Linearizability}.
Similarly to many sequential dynamic connectivity algorithms,
we use \emph{Euler Tour Trees} as the backbone structure, where each tree represents one connected component. Instead of making the ETT algorithm fully concurrent, we present a \emph{single-writer} solution which implements the \texttt{connected(..)} operation in a non-blocking way; hence, all the modifications should be guarded by a lock. 
In short, we partition \ett{} merge and split operations into their logical and physical components, so these operations appear atomic to concurrent readers, regardless of internal cuts and links. To make concurrent reads linearizable, we introduce versioning for the roots of the Cartesian trees~\cite{cartesian-trees} {---} which represent Euler Tour Trees; please see  Section~\ref{sec:ett} for a detailed overview. These versions help to detect connectivity changes when snapshotting roots for the connectivity check. The ideas we introduce can be employed for developing a concurrent single-writer version of any balanced tree which supports the standard \texttt{merge(A, B)}, \texttt{split(T)}, \texttt{add(u)}, \texttt{remove(u)}, and \texttt{same\_tree(u, v)} operations. 

Using the algorithm of Holm et al.~\cite{Holm2001} for the dynamic connectivity problem and these single-writer \ett{}-s, we automatically get a single-writer dynamic connectivity algorithm.
We find the ``single-writer'' restriction reasonable in this application, since it allows us to implement Euler Tour Trees almost as efficiently as in the sequential case, while allowing  queries, which are usually the most frequent operations, to run in parallel. 
Furthermore, this assumption reduces the overall complexity compared to a potential fully lock-free algorithm, which is bound to involve extremely complex book-keeping of components. 
Specifically, imagine a fully lock-free generalization of the  algorithm  of Holm et al. for the dynamic connectivity problem, e.g., implemented via lock-free software transactions. Since a replacement search takes $\Omega(N)$ in the worst case, and other operations can cause it to fail and re-run, the upper bound on the total work would be $\Omega(MN)$, which is substantially worse than $\O(M \log^2 N)$ of the sequential solution, where $M$ is the number of operations. 

Next, to improve scalability, we discuss how fine-grained locking can be applied to our proposal, so that updates on different components can be naturally parallelized. Surprisingly, this variant also requires our single-writer \ett{}, even when all operations, including connectivity queries, are executed under locks. 
A further optimization we propose gives a way to perform edge additions and removals that do not change the spanning forest in a non-blocking way. 

Finally, we evaluate the resulting algorithm on various workloads, both on real-world and synthetic graphs. The results show the efficiency of each of the proposed optimizations; the most efficient variant improves the performance of a coarse-grained locking implementation on realistic scenarios up to 6x on average and up to 30x when connectivity queries dominate the workload. 

This work has been published at SPAA '21 and is available as~\cite{spaa-publication}.

\section{Related Work}
\paragraph{Sequential Algorithms}
Variants of the connectivity problem have a rich history in the sequential algorithm literature, e.g.~\cite{tarjan1975efficiency, tarjan1979class, shiloach1981line}. 
The first polylogarithmic-time dynamic connectivity algorithm was proposed by Henzinger and King~\cite{Henzinger1999}. The main idea is to maintain a  spanning forest using \emph{Euler Tour Trees (ETTs)} 
and partition the rest of the graph into $\O(\log N)$ levels. We discuss Euler Tour Trees in detail in Section~\ref{sec:ett}. 
Whenever removing spanning edges, they employ a \emph{random sampling} technique to find an edge that can reconnect the trees. As a result, updates achieve randomized $\O(\log^3 N)$ work on average per operation, while \connected{} queries require only $\O(\frac {\log N} {\log \log N})$ work.

A faster and arguably simpler \emph{deterministic} algorithm was introduced by Holm, Lichtenberg, and Thorup~\cite{Holm2001}. It is based on the same ideas as the algorithm of Henzinger and King and uses nested levels of edges to get deterministic amortized $\O(\log^2 N)$ per operation and the same $\O(\frac {\log N} {\log \log N})$ time complexity for reads. Later, Thorup decreased the amount of used memory to $\O(N+M)$~\cite{Thorup2000}, where $M$ is the number of operations.

Further algorithms proposed additional enhancements. First, in 2011, Wulff-Nilsen  introduced the \emph{shortcut structure} and \emph{lazy local trees} to optimize modification time complexity by a $\log \log N$ factor~\cite{Wulff-Nilsen2013}. Second, in 2015, Huang et al. applied these ideas to the Henzinger and King algorithm, achieving $\O(\log N (\log \log N)^2)$ amortized modification time and $\O(\frac {\log N} {\log \log \log N})$ query time.

As for the algorithms optimizing worst-case time, the best known algorithm was proposed by Wang in 2015~\cite{Wang2015}. It employs the \emph{cutset structure} to get $\O(\log^4 N)$ modification time with an exponentially small probability of an error. The existence of a deterministic algorithm with polylogarithmic worst-case time per operation is still open. The best such algorithm has $\O\left(\frac {\sqrt N} {w^{1/4}}\right)$ time by combining the \emph{square root decomposition} and \emph{bit compression}, where $w$ is the machine word size.

Patrascu and Demaine proved that the lower bound on operation time complexity is $\Omega(\frac {\log N} {\log \log n})$~\cite{Patrascu2005}. Furthermore, Patrascu and Thorup showed that the modification time can not be $o(\log N)$, unless the query time is $\Omega(N^{1 -o(1)})$~\cite{Patrascu2011}. Consequently, all presented polylogarithmic algorithms can be seen as ``nearly-optimal.''

\paragraph{Parallel Algorithms}
Acar et al. in 2019 suggested employing parallelism to process batches of operations of the same type~\cite{Acar2019}. Their algorithm has no total work overhead due to multiple threads and becomes asymptotically faster with larger batch sizes, however, even for batches of size $n^{1-\varepsilon}$ there still will be no asymptotic difference, where $\varepsilon > 0$ is any constant. Another problem is the requirement to have large batches of queries of the same type. Meanwhile, if there is a \emph{concurrent} algorithm, it is trivial to use it to process batches in parallel, which means that a \emph{concurrent} algorithm can solve more general problems.

To construct a concurrent algorithm, Aksenov et al. applied the \emph{flat combining} technique~\cite{flat-combining} to the Holm et al. algorithm in~\cite{Aksenov2018}. However, this technique's standard implementation does not provide much better performance than the coarse-grained locking algorithm. Therefore, they introduced its read-optimized version named \emph{parallel combining}, which processes query operations in parallel and scales better than the standard flat-combining on read-dominated scenarios.

A concurrent lock-free algorithm was proposed by Chatterjee et al.~\cite{Chatterjee2019}. They develop a BFS-based connectivity query with linear time complexity per operation, which is inferior to the polylogarithmic time of the sequential algorithms. 

As far as Euler Tour Trees are concerned, Tseng et al. designed a batch-parallel ETT algorithm using a phase-concurrent skip-list~\cite{Tseng2019BatchParallelET}. \emph{Phase-concurrent} means that no concurrent operations of different types are allowed. However, in our algorithm we focus on allowing concurrent reads and updates  because reads are fast and frequent, while modifications are slow and rare in realistic read-dominated scenarios. Intuitively, in such scenarios, slow and rare operations should not cause fast and frequent operations to wait.

Compressed functional trees with merge and union operations were used by Dhulipala et al. for low-latency graph streaming~\cite{Dhulipala2019LowlatencyGS}. They get concurrent reads ``for free'', since they employ \emph{immutable} data structures. Similarly to Tseng et al.~\cite{Tseng2019BatchParallelET} and Acar et al.~\cite{Acar2019}, they use parallelism to process batches of operations. Nevertheless, in our problem trees should have links both from parents to children and from children to parents, because connectivity queries traverse trees from bottom to top. This makes batch-parallel processing and partially-persistent data structures much more difficult, so the same ideas can not be applied easily to the dynamic connectivity problem.

\section{Single-Writer Euler Tour Tree}\label{sec:ett}

\setlength{\columnsep}{5pt}
\setlength{\intextsep}{5pt}
\begin{wrapfigure}{r}{0.2\textwidth}
% \vspace{-1em}
\includegraphics[width=0.21\textwidth]{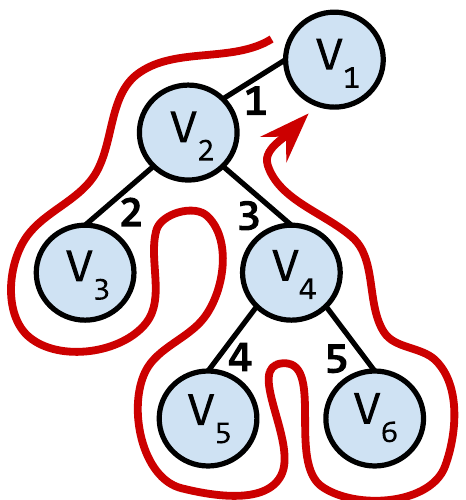}
\end{wrapfigure}

As a first step, we present a concurrent single-writer multiple-reader lock-free \emph{Euler Tour Tree (ETT)} algorithm. 

\paragraph{Sequential Euler Tour Tree}
In short, \ett{} specifies an undirected rooted tree as a graph traversal that starts and ends in the root and covers all edges exactly twice {---} see the picture on the right as an example. This traversal is exactly the same as in the depth-first search (DFS) graph algorithm started from the root.

Henzinger and King solved the dynamic connectivity problem by representing each component via an \ett{}~\cite{Henzinger1999}. Consider the traversal path from the picture above, which stores all the edges and the first occurrence of each vertex: 

\begin{center}
  \textbf{\texttt{[\color{red}$v_1$,1,$v_2$,2,$v_3$,2\color{black},3,\color{blue}$v_4$,4,$v_5$,4,5,$v_6$,5\color{black},3,\color{red}1\color{black}]}}  
\end{center}

Assume we want to split the original tree by removing the edge ``3'' and produce two \ett-s as a result. The edge ``3'' cuts the traversal path into three parts.
Notice that the blue part in the center and the linked red parts on the ends produce the \ett-s we need. Thus, we can store these traversals in a data structure with \texttt{cut} and \texttt{link} operations, so split is performed by two \texttt{cut}-s and one \texttt{link}. For reasons described later, we use Cartesian trees~\cite{cartesian-trees} with implicit keys, where every node has a random priority, and the tree forms a heap according to these priorities. This data structure supports \texttt{cut} and \texttt{link} operations in $\O(\log N)$ expected time. 

This way, \ett-s solve the dynamic connectivity problem for forests (sets of trees) in the following way:
\begin{itemize}
    \item \texttt{connected(u, v)} finds the corresponding Cartesian tree roots starting from the nodes \texttt{u} and \texttt{v} and checks them for equality;
    \item \texttt{add\_edge(u, v)} merges the corresponding Cartesian trees;
    \item \texttt{remove\_edge(u, v)} splits the Cartesian tree in the way described above.
\end{itemize}
All these operations work in $\O(\log N)$ time.

\paragraph{Atomic Merge and Split}
We want to construct a linearizable \ett{} implementation that allows \texttt{add\_adge} and \texttt{remove\_edge} operations to be executed by only one thread and the \texttt{connected} query to be executed by multiple threads in a non-blocking way. The idea is to use the same sequential algorithm, when \texttt{connected(u, v)} finds the roots of \texttt{u} and \texttt{v} in the corresponding Cartesian trees. However, we need to do both searches atomically under concurrent updates. 

The first problem on the way to such an implementation is the appearance of out-of-thin-air components during \ett{} splits and merges. 
Consider an edge addition between \texttt{AB} and \texttt{CD} components on Figure~\ref{merging}. To merge the corresponding \ett{}-s, the Cartesian trees should be cut by the edge endpoints and then linked in the presented way. However, when \texttt{connected(..)} is in the middle of the update and the cuts are finished but the trees are not linked yet, concurrent readers may see non-existing connectivity components, appeared out of nowhere.

\begin{figure}[h]
\centering
\includegraphics[width=0.9\linewidth]{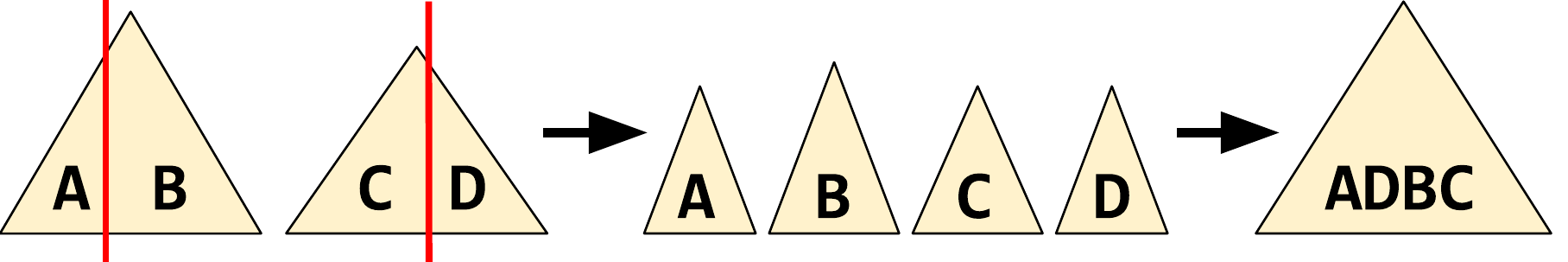}
\caption{Edge addition between \texttt{AB} and \texttt{CD} components. The corresponding Cartesian trees should be cut by the edge endpoints, followed by linking the resulting four parts in a proper order. However, in the center state concurrent readers may see non-existing connectivity components.}
\label{merging}
\end{figure}

In order to solve this problem, we partition \ett{} splits and merges into \emph{logical} and \emph{physical} parts. The logical part consists only of one instruction and will be the linearization point of the operation.

The edge addition (merge) procedure starts with linking one Cartesian tree to another by setting the parent link {---} this way; the trees seem already merged for concurrent readers. It is correct since  for the Cartesian trees we know the node that will become the common root -- the root that has higher priority. Note that parent links always lead to a node with higher priority, so the parent link graph is acyclic and there is only one sink node after the first step. The following cuts keep the parent links, so the tree stays merged during the restructuring. Figure~\ref{linearizable-merging} illustrates this algorithm.

% To implement non-blocking concurrent reads we should start not with read operations but with modifications. The problem is that edge additions and removals in Euler tour trees are not atomic. For example, see Figure~\ref{merging}. Upon new edge addition, Euler tours of the components are firstly split into parts by the positions of the edge endpoints and then they are linked in a certain order. The data structure state after the first step is not linearizable, because at no moment of time there should be four different components of connectivity. The ideas of making modifications atomic will also be vital for constructing fine-grained algorithm over components of connectivity in the concurrent dynamic connectivity algorithm.

\begin{figure}[h]
\centering
\includegraphics[width=0.6\linewidth]{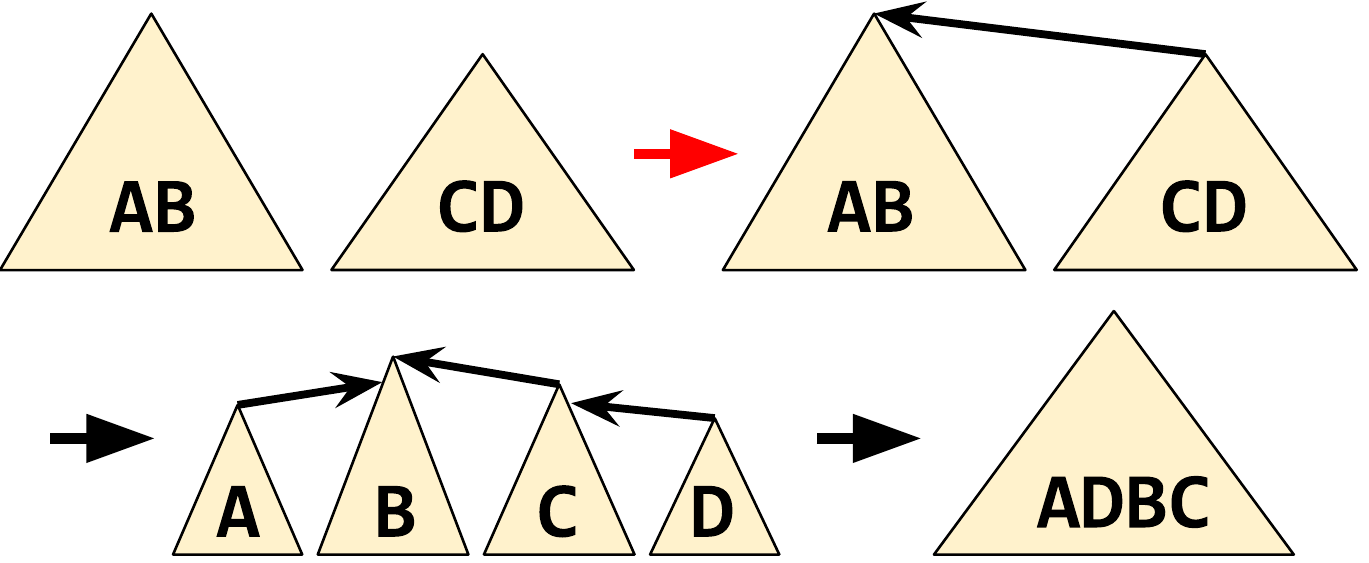}
\caption{\emph{Atomic} edge addition between \texttt{AB} and \texttt{CD} components. At first, the corresponding trees are merged logically. The tree should then be restructured, but the parent links of the subtrees produced by cuts are kept.}
\label{linearizable-merging}
\end{figure}

The edge removal (split) procedure first prepares the tree to be split and then applies the operation by a single parent link change. It means that we at first cut the Cartesian tree into three parts, $A$, $B$, and $C$, but keep the parent link as we did for merging. After that, the first and the third trees, $A$ and $C$, are united. The constructed $AC$ and $B$ trees correspond to the \ett{}-s we need to have at the end. Notice that either the parent of the root of $B$ is a node in $AC$ or vice versa. The corresponding parent link removal splits the \ett{}-s both physically and logically. Figure~\ref{linearizable-splitting} illustrates the described procedure.

\begin{figure}[h]
\centering
\includegraphics[width=0.6\linewidth]{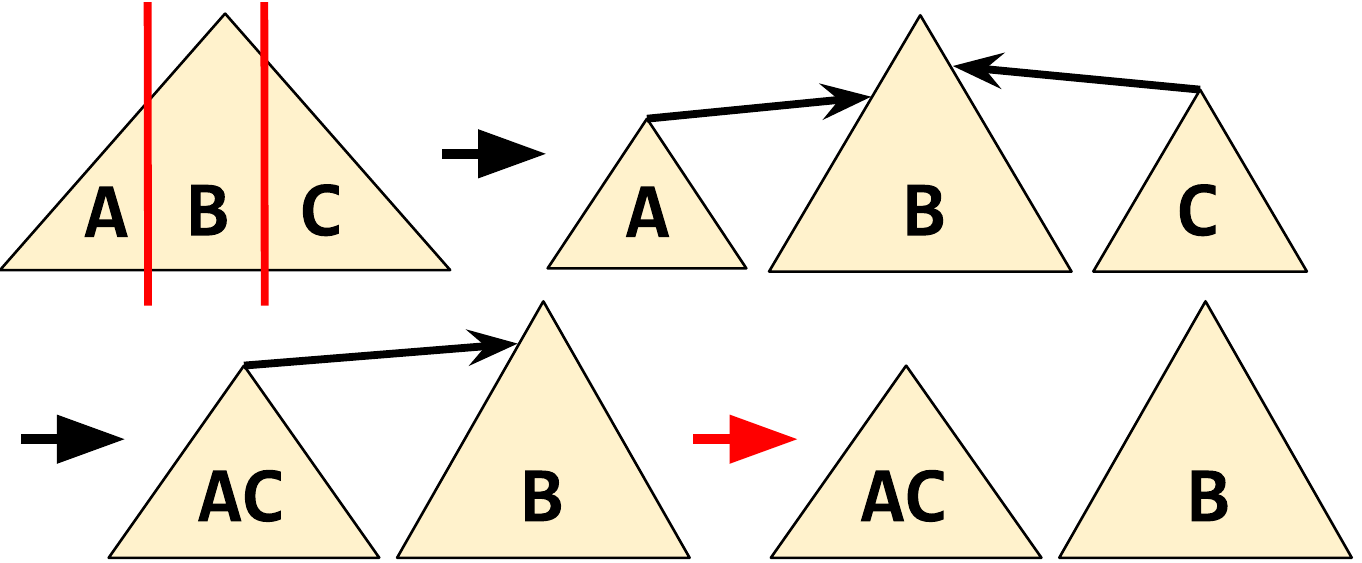}
\caption{\emph{Atomic} edge removal. The algorithm keeps the parent links during the required cuts and merges, and the last unlinking applies the operation both logically and physically.}
\label{linearizable-splitting}
\end{figure}

% To solve the problem we partition ETT modifications into two parts -- a logical part that consists of a single atomic operation and a clean-up part. The logical part will be the linearization point of these operations. In the case of edge addition, the first step and the logical part is creating a parent link from one root to the future common root (Figure~\ref{linearizable-merging}). In Cartesian trees, the root of the result of linking trees is one of the roots of linked trees, specifically the root with higher priority. Then it is important to preserve this parent link while proceeding as in the sequential algorithm. Since after the logical part there is only one node that does not have a parent link and all parent links lead from nodes with lesser priorities to nodes with greater priorities, all readers in a connected component will come to the same root. As a result, after the logical part, the operation will be considered completed by concurrent readers.

% A similar method can be applied to edge removals. However, in contrast to edge additions, the linearization point will be the last step rather than the first (Figure~\ref{linearizable-splitting}). In this operation, the Euler tree is firstly split into three parts by positions of two entries of the edge. Note that as in the edge addition parent links of roots are preserved. After that, the first and the last trees are merged. The last step and the logical part is the deletion of the parent link between two trees, which now represent different components of connectivity.

\paragraph{Linearizable \texttt{connected(u, v)} Operation}
Once the out-of-thin-air problem is solved, we can construct a linearizable algorithm for \texttt{connected(u, v)}. The idea is to maintain versions of the Cartesian tree roots and change them on every split and merge operation. We want to use these versions to check that roots for \texttt{u} and \texttt{v} were snapshotted correctly by re-checking these versions.

Yet, this approach immediately meets a problem {---} it is non-trivial to perform the logical part of the operation and change the root version atomically since these actions change different memory locations.
For instance, when an edge addition changes the parent link that belongs to the root with lower priority, the version increment should be performed on the root with higher priority (see Figure~\ref{linearizable-merging} that illustrates the edge addition procedure).

As a solution, we suggest updating the root version \emph{before} each merge and split procedure; thus, before each edge addition and removal. Since there is a single writer, it is guaranteed that the version is \emph{at most} one step ahead for concurrent readers. However, this relaxation requires a non-trivial check to ensure that the obtained roots were snapshotted atomically. 

If version and parent link updates were atomic, then the connectivity check would find the Cartesian tree roots for both $u$ and $v$, and then find the root for $u$ again. When both roots and the corresponding versions found for $u$ coincide, it is guaranteed that the snapshot is correct, and the connectivity check should simply compare the obtained roots.

With the versions that can be one step ahead, we may need more re-checks. When roots for $u$ and $v$ coincide and the second root search for $u$ returns the same node and version, there was no change between the first root search for $u$ and the root search for $v$ that affected the connectivity check. 

However, when the obtained roots for $u$ and $v$ are different, and the second root search for $u$ returns the same node and version, we cannot be sure that $u$ and $v$ are in different components.
Despite this observation, an additional pair of checks that the roots of $u$ and $v$ and their versions has not been changed solves the problem. The corresponding proof is discussed later in this section. In addition, we show that the algorithm is incorrect without the last root search and the corresponding re-check in Appendix~\ref{appendix:ett_counterexample}.

The described algorithm is presented in Listing~\ref{nonblocking-reads}.

\begin{figureAsListing}
  \begin{lstlisting}
func connected(u, v: Vertex): Bool = while (true) {
  // Get the roots of `u` and `v` with the versions
  (u_root, u_version) := find_root(u)
  (v_root, v_version) := find_root(v)
   // Has the component of `u` been changed?
  if (u_root, u_version) != find_root(u): continue
  if u_root != v_root: 
  #\indentrule#  // `u` and `v` are likely to be in different 
  #\indentrule#  // components, re-check that `u_root` and
  #\indentrule#  // `v_root` were snapshotted atomically
  #\indentrule#  if (v_root, v_version) != find_root(v): continue    
  #\indentrule#  if (u_root, u_version) != find_root(u): continue#\label{line:connected:last_search}#
  return u_root == v_root
} 
// Returns the root with the current version
func find_root(u: Vertex): (Vertex, Int) {
  cur := u
  while (true):
  #\indentrule#  parent := cur.parent // read the parent of `cur`
  #\indentrule#  if parent == null: break // is `cur` the root?
  #\indentrule#  cur = parent // go towards the root
  // Return the found root with the current version
  return (cur, cur.version)
}
  \end{lstlisting}

  \caption{Linearizable connectivity check. \texttt{find\_root} follows parent links until it finds the root and its version.}
  \label{nonblocking-reads}
\end{figureAsListing}

\paragraph{Proof of Correctness}
To prove the correctness of the \texttt{connected(u, v)} algorithm we first prove the following lemma:
\begin{lemma}
Assume that \texttt{find\_root} returned $(r, v1)$ for some vertex and another \texttt{find\_root} returned $(r, v2)$. Consider the component of connectivity with vertex $r$. If a modification in it started after the first \texttt{find\_root} completion and finished before the root is found in the second \texttt{find\_root} invocation, then the root version changed, i.e. $v1 \neq v2$. By ``root finding'' we mean the moment when \texttt{find\_root} reads \texttt{cur.parent} that equals \texttt{null}.
\label{lemma-reads}
\end{lemma}
\begin{proof}
The case when the modification happened while $r$ was still the root of its component is trivial since this modification directly changed the version of $r$.

The only case left is when the root was another vertex. However, since we know that the second \texttt{find\_root} invocation returned $r$, there should be another modification that made $r$ the root again {---} it has to increase the version of $r$ at the beginning.
\end{proof}

\begin{theorem}
\label{reads-linearizable}
The connectivity check algorithm in Listing~\ref{nonblocking-reads} is linearizable.
\end{theorem}
\begin{proof} Proof by contradiction. There can be only two reasons why the connectivity check is not linearizable: either there was no point in time when $u$ and $v$ were in the same connected component, and the check returned \texttt{true}, or there was no point in time when $u$ and $v$ were in different components, and the check returned \texttt{false}.

\emph{Case 1: \texttt{connected(u, v)} returned \texttt{true}.} 
Assume that there was no point in time when $u$ and $v$ were in the same component. To succeed, all three \texttt{find\_root} invocations should find the same root $r$ with the same version.

Consider the connectivity of node $r$ at the moment immediately after the first root search for $u$. According to the assumptions, at this moment $r$ was either connected to only one of the nodes or not connected to both. In the first case, in order to make another node connected with $r$ and return it in the \texttt{find\_root} invocation, there should be at least one edge removal to split the connected node and $r$ and one edge addition to make a \texttt{find\_root} invocation for another, disconnected node return $r$. The second case is similar: there should be at least two edge additions in the same component of connectivity, the last one of which extends the component of connectivity with $r$.
As a result, in both cases there were two modifications, and the second one started after the first \texttt{find\_root} invocation. This way, there is a modification between the completion of the first root search for $u$ and the second one, both returned the same root $r$. 
Following the Lemma~\ref{lemma-reads}, the obtained versions should be different. Contradiction.

\emph{Case 2: \texttt{connected(u, v)} returned \texttt{false}.}
Assume that during this \texttt{connected(u, v)} invocation $u$ and $v$ were always in the same component of connectivity. 
In this case, the following history should be obtained:
\begin{center}
    \texttt{find\_root(u): $r_u,\ V_u$}
    
    \texttt{find\_root(v): $r_v,\ V_v$}
    
    \texttt{find\_root(u): $r_u,\ V_u$}
    
    \texttt{find\_root(v): $r_v,\ V_v$}
    
    \texttt{find\_root(u): $r_u,\ V_u$}
\end{center}
% To get the result \texttt{false}, the first search must have found root $r$ for the vertex $u$, then root $w$ for the vertex $v$, after it re-find root $r$ for $u$, root $w$ for $v$, and finally again find root $r$ for $u$. All found versions should coincide respectively too.   

Remember that according to the assumption $u$ and $v$ have always been in the same component. Therefore, since four root changes of the same component have been detected, there were four corresponding modifications. The third of these modifications started after the first \texttt{find\_root} invocation and influenced the result of the fourth, thus we can use it to apply the Lemma~\ref{lemma-reads} to the first and the last \texttt{find\_root(u)} invocations and get a contradiction.

% Similarly to the first case, we will apply Lemma~\ref{lemma-reads} to root $r$, the first and the last searches of this root. What is left is to find the modification that happened after the first search and before the end of the last one. It is easy to see that for every \texttt{find\_root} result change in a component of connectivity there must have been a modification in the component. In our case, since both vertices were in the same component of connectivity by the assumptions, there were 4 such result changes and, therefore, 4 corresponding modifications. 

% Consider these modifications. The first one could have started and finished before the end of the first search if it changed the part of the path already covered by the first search. The second modification could have started before the first search, but must have finished after the first search, because it corresponds to the change between the second and the third searches. Since modifications are executed by a single writer, the third modification started after the second one, and thus, it started after the first search. More than that, it finished before the fifth search, because it corresponds to the change of result between the third and the fourth searches. As a result, this modification can be used in Lemma~\ref{lemma-reads} and the fifth search must have read an updated version leading to the contradiction.
\end{proof}

\paragraph{Other Balanced Trees}
We used the Cartesian trees as an underlying data structure in \ett-s because parent links in the Cartesian trees naturally form an acyclic graph at any moment of time and roots do not change due to rebalancing. Otherwise, concurrent readers could have looped endlessly or could have wrongly detected a component change during a modification. In fact, our single-writer multiple-reader implementation of the Cartesian trees is barely changed from the sequential implementation. The main issue with the Cartesian trees is that their time complexity is \emph{expected} due to randomization. So, we need to show how to employ any \emph{deterministic} balanced tree instead, for example, the B-tree~\cite{Btree}. 

In the B-tree each node contains a number of keys rather than a single key. In our case, a node can be viewed as a container of some vertices and edges from graph $G$. As a consequence, to allow root searches we also should store links from vertices to nodes they are contained in. The ideas for atomic merges and splits, and for linearizable queries are the same -- in a logical part of modifications one tree is linked/unlinked to/from another tree, each root has a version that changes before a modification. The only problem is that the B-tree can allocate a new node and make it the root during a merge or replace the root during a split, and we want roots not to change unless in the logica   l part. The solution is straightforward -- if the B-tree replaces the root with another node, we can instead swap the content of the root and the content of this node. Note that the structure of the tree remains the same, but concurrent readers now can not see odd root changes.

The number of keys in each B-tree node can be chosen arbitrarily; a B-tree with $\Theta(k)$ keys in each node supports \texttt{cut} and \texttt{link} operations in $\O(k \log_{k} N)$ time, and \texttt{find\_root} in $\O(\log_k N)$. As a result, if we choose constant $k$, we will get the same, but \emph{deterministic} asymptotic time complexities as for the Cartesian tree. What is more, $\Theta(\log n)$-ary trees can improve read time complexity for the general dynamic connectivity problem to $\O(\frac {\log N} {\log \log N})$ as shown by Henzinger and King~\cite{Henzinger1999}. Besides, nodes in other balanced binary trees can be viewed as containers of size $1$, and thus, we can adapt them in the same way as the B-tree by introducing this additional level of indirection.

\section{Concurrent Dynamic Connectivity}
We design a concurrent generalization for the classical sequential dynamic connectivity algorithm proposed by Holm et al.~\cite{Holm2001} for several reasons. First of all, this algorithm can be improved to the best known deterministic algorithm~\cite{Wulff-Nilsen2013} by supplementing it with the \emph{lazy local trees} and \emph{shortcut structure}. Given that this algorithm is already complex, it makes sense to start with making a generalization for it. Secondly, it uses the same ideas as the randomized algorithm with amortized time complexity. Moreover, we prioritize average time over worst-case time as it is a better metric for throughput.

\subsection{The Sequential Algorithm via \ett{}-s}
The classic algorithm that solves the dynamic connectivity problem uses Euler Tour Trees to represent different components of connectivity. The idea is to maintain a spanning forest $F$ of graph $G$, so most of the dynamic connectivity operations can be straightforwardly reduced to the \ett{}-s:

\begin{itemize}
    \item \texttt{add\_edge(u, v)} either merges two \ett{}-s if it connects different components, or adds $(u, v)$ to the set of non-spanning edges. If the edge $(u, v)$ connects different components of connectivity, then it should be added to the spanning forest, i.e. Euler Tour Trees.
    \item \texttt{remove\_edge(u, v)} either removes $(u, v)$ from the set of non-spanning edges or search for a \emph{replacement} edge -- an edge that reconnects the spanning trees after the spanning $(u, v)$ edge is removed, reconstructing the \ett{} if a replacement is found or splitting it otherwise.
    \item \texttt{connected(u, v)} finds the roots in the Cartesian trees of  vertices $u$ and $v$ and compares them.
\end{itemize}

% solves the Dynamic Connectivity Problem in amortized $\O(\log^2 n)$ time per operations by employing the ETT for maintaining a spanning forest $F$ of graph $G$. Most :
% \begin{itemize}
    % \item \texttt{connected(u, v)}. Finds the roots in the Cartesian trees of the vertices $u$ and $v$ and compares them.
    % \item \texttt{add\_edge(u, v)}. If the edge $(u, v)$ connects different components of connectivity, then it should be added to the spanning forest, i.e. Euler Tour Tree. Otherwise, the spanning forest does not change.
%     \item \texttt{remove\_edge(u, v)}. If the edge $(u, v)$ does not belong to the spanning forest, then we can just remove it. Otherwise, we need to search for a \emph{replacement} edge -- an edge that reconnects the spanning trees after the edge removal. A full scan of all edges leads to a linear time complexity, so instead Lichtenberg et al. introduce a \emph{level structure}.
% \end{itemize}

\paragraph{The Level Structure}
Note that a full scan of all non-spanning edges in \texttt{remove\_edge} leads to a linear time complexity. To solve this problem, Holm et al. introduced a special \emph{level structure} for storing non-spanning edges to search for replacement efficiently.

Each edge $e$ is assigned with an integer non-negative level $l(e) \leq l_{max} = \lfloor \log_2 n \rfloor$. Edges of level $\geq i$ form the subgraph $G_i$. In particular, $G = G_0 \supseteq G_1 \supseteq \ldots \supseteq G_{l_{max}}$. In each subgraph $G_i$ a spanning forest $F_i$ is maintained with Euler Tour Trees. 

Furthermore, the level structure supports the following invariants:
\begin{itemize}
    \item $F = F_0 \supseteq F_1 \supseteq \ldots \supseteq F_{l_{max}}$. Intuitively, this means that to get the spanning forest of $G_i$ we can take $F$ and omit edges of the levels lower than $i$.
    \item $F$ is a maximal spanning tree for edge weights defined as their levels.
    \item Every component of connectivity in $G_i$ has size at most $n / {2^i}$. As a result, edges cannot have levels greater than $\lfloor \log_2 n \rfloor$.
\end{itemize}
The main idea of how to obtain better amortized operation time is to increase the level of every edge processed during a replacement edge search if it can not be a replacement. As the maximum level is bounded with $\lfloor \log_2 n \rfloor$, the total number of level increases is at most $m \log_2 N$, where $m$ is the number of inserted edges.

The \texttt{connected(u, v)} and \texttt{add\_edge(u, v)} operations in this scheme work with the lowest level $G_0 = G$ and are not changed relatively to the general description above. The information about non-spanning edges (the edges that do not belong to the spanning forest) on the level $l$ is stored in adjacent nodes of the level $l$. Specifically,  each node of the Cartesian tree has a set of non-spanning edges. The \texttt{remove\_edge(u, v)} operation when removing a non-spanning tree of level $l$ just deletes the information about it in the corresponding nodes of $F_l$.

Upon a spanning edge of level $l$ removal, a replacement edge is first searched among the edges of level $l$. For this purpose, spanning tree $T$ is split into trees $A$ and $B$ that will appear if there is no replacement edge. Without loss of generality, we can assume that the size of $A$ is not greater than the size of $B$. Then, all spanning edges on the level $l$ in $A$ increase their levels by $1$. It does not violate the invariant on the sizes of connected components because $A$ is at least twice less than the component of connectivity $T$ and the size of $T$ was at most ${n} / {2^l}$. After that, the algorithm iterates over all non-spanning edges of level $l$. Those edges that can not be a replacement increase their levels as well. They still will be non-spanning in $G_{l+1}$ due to the previously mentioned spanning edge rises. If an edge can be a replacement, then the search is finished and the edge is added to the spanning forest $F_i$ for every $i \leq l$. If all edges are covered and no replacement is found, the spanning tree $T$ stays split into $A$ and $B$, and the search continues in the lower levels, until $l \geq 0$. 

To iterate over all spanning or non-spanning edges of the current level effectively, each node of the Cartesian tree has flags whether there are such edges in its subtree. This way, a new edge on the current level can be found in $\O(\log N)$ time by descending from the root to children with such edges.

\subsection{Lock-Free Connectivity Checks}
We start constructing our scalable concurrent algorithm for dynamic connectivity iteratively. At first, assume that all edge additions and removals are ordered by a coarse-grained lock or using the flat combining technique.

However, with the lock-free single-writer concurrent \ett{} algorithm presented in Section~\ref{sec:ett}, we immediately get a non-blocking connectivity check for the dynamic connectivity problem by replacing a sequential \ett{} for the lowest level $G_0$ with the concurrent one. This replacement is correct since connectivity queries work only with the lowest level of the level structure, while a single writer, in fact, performs all modifications. Besides, if a $\Theta(\log n)$-ary balanced tree is employed for the lowest level, query time complexity can be improved to $\O(\frac {\log n} {\log \log n})$ without affecting average modification time~\cite{Henzinger1999}.

\subsection{Fine-Grained Locking for Updates}
Intuitively, modifications of distinctive components can be naturally parallelized.
That leads to an idea to replace the global lock with fine-grained per-component locking.
Nonetheless, a problem has to be solved:  the components of connectivity change dynamically, so we need a representative of each component that can hold the lock.

% To allow modifications in different components of connectivity to be executed in parallel we replace coarse-grained locking with fine-grained locking on components of connectivity. This means that before the actual modification, an operation takes locks of all components of connectivity it wants to change.

Similarly to the non-blocking reads, connectivity components can be represented by the Cartesian tree roots of the \ett{}-s of $G_0$. 
This way, an update operation can obtain the roots via \findRoot{} invocations and take the corresponding locks. 
Once the locks are acquired, we have to check that the found nodes are still roots, i.e the operation locked \emph{some} components. After that, we repeat the root searches for both $u$ and $v$. 
Since after locking the component of the root can not be modified, when \findRoot{} returns the same root, it is guaranteed that the correct representative is acquired, and no operation can modify the corresponding component while we hold the lock.

Listing~\ref{lst:fg} presents a pseudo-code of the algorithm that locks the required components. The rest of the design stays the same as for the coarse-grained locking.

% a component of connectivity in the ETT of  can act as a representative of the component. Then, locking a component of connectivity is taking the lock in its root. Since roots are not known beforehand, the operation needs to find them by following parent links, starting from vertices given in the query. On account that parallel threads can modify the components of connectivity during the search, after taking the lock, the operation checks that it is the correct lock by repeating the root search. When the roots coincide, the  If the repeated search shows that the incorrect lock is taken, the operation releases the lock and repeats the process again. To avoid deadlocks due to fine-grained locking we employ the technique known as \emph{lock hierarchy}. In other words, the operation first finds both roots and only then takes the corresponding locks according to a predefined order.

% Another problem of the fine-grained lock is that the modifications in the sequential ETT are not atomic, because a root search can see odd non-existent components of connectivity (Figure~\ref{merging}). However, this problem was already solved in the previous section by linearizing modifications in the concurrent ETT, and furthermore, we already replaced the sequential ETT for $G_0$ with a concurrent one.  

\begin{figureAsListing}
  \begin{lstlisting}
func lock_components(u, v: Vertex) = while (true) {
  // Get the roots of `u` and `v` with the versions
  (u_root, _) := find_root(u)
  (v_root, _) := find_root(v)
  // Always have the same lock ordering
  if u_root < v_root: swap(u_root, v_root)
  lock(u_root), lock(v_root) 
  // Re-check that the correct locks were acquired
  if u_root.parent != null or
     v_root.parent != null or
     u_root != find_root(u).root or
     v_root != find_root(v).root:
  #\indentrule#  unlock(u_root), unlock(v_root) // unlock 
  #\indentrule#  continue                       // and restart
  return // the locks were taken correctly
}    
  \end{lstlisting}
  \caption{The per-component fine-grained synchronization for update operations that works with vertices $u$ and $v$. The Cartesian tree roots, returned by \findRoot{}, represent the components.}
  \label{lst:fg}
\end{figureAsListing}

\subsection{Lock-Free Non-Spanning Edge Updates}\label{subsec:non_spanning_simple}
Given that at most $N - 1$ edges can belong to a spanning forest and there can be up to $\O(N^2)$ different edges in a graph, most modifications in dense graphs do not change the spanning forest. For example, in the real-world hyperlink and physical connections between routers graphs, the ratio of edges over vertices often exceeds $10$~\cite{realgraphdensity}. Therefore, about $90\%$ edges of the graph are non-spanning. The optimization proposed in this subsection performs non-spanning edge additions and removals without taking locks.

\begin{figureAsListing}
  \begin{lstlisting}
#\label{line:class-edge}#class Edge(u: Vertex, v: Vertex, status: Status)
class SpanningEdgeRemoveOp(edge: Edge) {
  replacement: Edge? = null // `DONE` when done
}  

#\label{line:multiset}#val non_spanning_edges = ConcurrentMultiSet<Edge>()
#\label{line:remove-op}#var remove_op: SpanningEdgeRemoveOp? = null

#\label{line:remove-edge-start}#func remove_edge(edge: Edge) = while(true) {
  when edge.status {
  #\indentrule#  NON-SPANNING -> {
  #\indentrule#  #\indentrule#  // Update the edge status, re-start on failure
  #\label{line:remove-logical}##\indentrule#  #\indentrule#  if !CAS(&edge.status, NON-SPANNING, REMOVED): continue
  #\indentrule#  #\indentrule#  // Remove the edge from the set of non-spanning
  #\label{line:remove-physical}##\indentrule#  #\indentrule#  non_spanning_edges.remove(edge)
  #\indentrule#  #\indentrule#  return
  #\indentrule#  }
  #\indentrule#  SPANNING -> {
  #\indentrule#  #\indentrule#  // Take the lock and remove the edge
  #\label{line:blocking-remove-inv}##\indentrule#  #\indentrule#  blocking_remove_edge(edge)
  #\indentrule#  #\indentrule#  return
  #\indentrule#  }
  #\label{line:removed-return}##\indentrule#  REMOVED -> return // already removed
  #\indentrule#  INITIAL -> return // not added yet #\label{line:nonspanning:initial-return}#
  #\indentrule#}
#\label{line:remove-edge-finish}#}

func add_edge(edge: Edge) = while(true) {
  #\label{line:add-connectivity-check}#if connected(edge.u, edge.v):
  #\indentrule#  // Take the lock and add the edge
  #\label{line:add-blocking-inv}##\indentrule#  blocking_add_edge(edge)
  #\indentrule#  return
  // Try to add a non-spanning edge; status is INITIAL 
  #\label{line:add-optimistical}#non_spanning_edges.add(edge) // add to the multi-set
  // Is there a concurrent remove for which the edge can be a replacement?
  #\label{line:add-notice-start}#cur_remove_op := remove_op
  #\label{line:start}#if cur_remove_op != null and
    #\label{line:add-notice-end}##\label{line:can-be-repl}# cur_remove_op.can_be_replacement(edge):
  #\label{line:add-try-set-repl}##\indentrule#  if CAS(&cur_remove_op.replacement, null, edge):
  #\indentrule#  #\indentrule#  // The edge became the replacement
  #\label{line:add-multiset-remove}##\indentrule#  #\indentrule#  non_spanning_edges.remove(edge)
  #\label{line:add-status-update}##\indentrule#  #\indentrule#  CAS(&edge.status, INITIAL, SPANNING)
  #\indentrule#  #\indentrule#  return
  #\label{line:add-remove-complete}##\indentrule#  else if cur_remove_op.replacement == DONE:
  #\label{line:add-remove-complete-start}##\indentrule#  #\indentrule#  non_spanning_edges.remove(edge)
  #\indentrule#  #\indentrule#  // The edge is about to become spanning, perform the addition under the lock
  #\label{line:add-remove-complete-end}##\label{line:end}##\indentrule#  #\indentrule#  blocking_add_edge(edge)
  #\indentrule#  #\indentrule#  return
  // Try to make the edge non-spanning and 
  // check that its endpoints are in the same component
  #\label{line:add-recheck-start}#if !connected(edge.u, edge.v): 
  #\label{line:add-recheck-end}##\indentrule#  non_spanning_edges.remove(edge), continue
  #\label{line:add-final-update}#if CAS(&edge.status, INITIAL, NON-SPANNING): return // success
  else:
  #\label{line:add-failure-start}##\indentrule#  // A concurrent writer uses the edge as a
  #\indentrule#  // replacement and moved the status to SPANNING
  #\label{line:add-failure-end}##\indentrule#  non_spanning_edges.remove(edge)
  #\indentrule#  return
}    
  \end{lstlisting}
  \caption{Lock-free algorithm for non-spanning edge additions and removals in the simplified case, when all spanning edge updates use coarse-grained locking, non-spanning edges do not organize the level structure and are stored in a single multi-set. In addition, in this code we assume that threads cannot add the same edge concurrently.}
  \label{lst:non_spanning_simple}
\end{figureAsListing}

\paragraph{High-Level Idea}
The main idea is to assign statuses to the edges and process the ones that are marked as non-spanning in a non-blocking way. Thus, each update operation checks whether the edge is spanning or should become one and processes it under the corresponding locks in this case. Otherwise, the edge is non-spanning, and the operation can be performed in a non-blocking way, re-starting from the beginning if some intermediate check or atomic update fails. The intuition is that such non-blocking operations linearize at the points of status updates.

The principal challenge of non-blocking edge addition lies in a possible conflict of this operation with a concurrent spanning edge removal. If a spanning edge removal searches for replacement and does not find it, while in parallel a non-spanning edge that can be the replacement is added, the data structure state becomes non-linearizable after the component of connectivity split by the removal operation. The conflict indicates the need for an additional channel of communication between spanning edge removals and non-spanning edge additions. To create such a channel, the removal operation at the beginning publishes information about itself together with an initially empty slot for a replacement edge. This allows a concurrent non-spanning addition to propose its edge as a replacement.

\paragraph{Assumptions}
The whole algorithm for non-blocking modifications does not fit into space requirements and contains a ton of technical details, so we present a simplified solution that relies on a few assumptions and shows the main idea of the algorithm. At first, we assume that threads cannot add the same edge concurrently. 
Besides, we offer a solution that uses coarse-grained locking for updates and does not leverage the level structure from the Holm et al. algorithm {---} we maintain a global set of non-spanning edges instead. We consider applying the presenting approach to the fine-grained locking and the Holm et al. algorithm as a technical detail. However, we fix these assumptions in our implementation, which pseudo-code is presented in Appendix~\ref{appendix:non-spanning}. Note that we also do not have these assumptions in our experiments.

\paragraph{Edge Statuses}
Under the discussed assumptions, we need four statuses, the state machine of which is presented in Figure~\ref{status-diagram}:
\begin{itemize}
    \item \texttt{INITIAL} -- the initial status assigned to every new edge;
    \item \texttt{SPANNING} -- the edge is in the spanning forest, so its removal should be executed under the lock;
    \item \texttt{NON-SPANNING} -- the edge is not in the spanning forest, so its removal can be non-blocking;
    \item \texttt{REMOVED} -- the edge is removed.
\end{itemize}
In the complete solution, an additional \texttt{ADDING-SPANNING} status is required to synchronize concurrent additions of the same edge. Otherwise, an addition operation can not distinguish cases when an edge was already added to the spanning forest and when it is being added by a concurrent addition, and thus, can not decide whether it should synchronize with a parallel thread to be sure that the edge is added.

\begin{figure}
\centering
\includegraphics[width=0.6\linewidth]{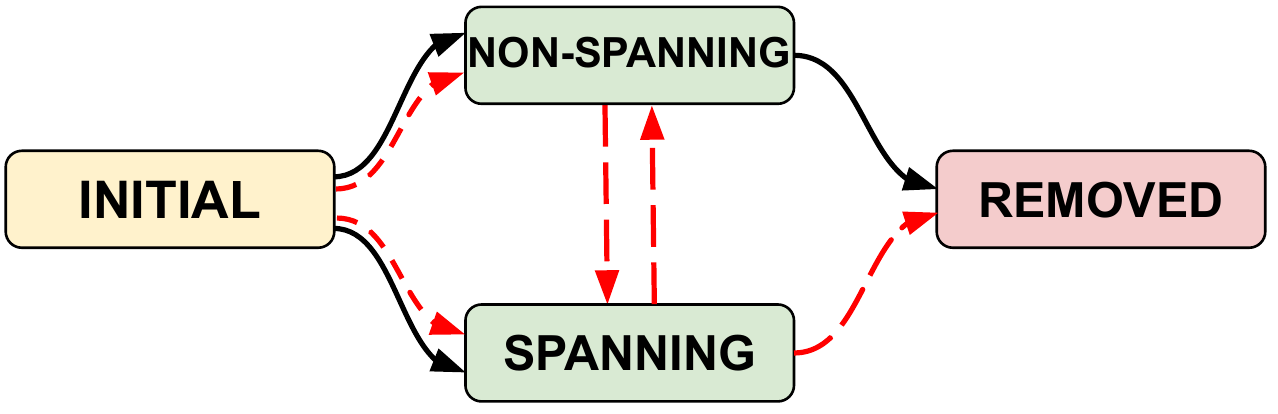}
\caption{Possible transitions between edge statuses. Red dashed transitions are performed under the lock.}
\label{status-diagram}
\end{figure}

\paragraph{The Data Structure}
Listing~\ref{lst:non_spanning_simple} presents an algorithm for non-blocking non-spanning edge processing. For simplicity, each edge is wrapped with a special \texttt{Edge} class that also stores the status from Figure~\ref{status-diagram} (line~\ref{line:class-edge}). In a practical implementation, one may use a concurrent hash table to store the statuses. What is more, we have \texttt{REMOVED} status logically, but in practice, it can correspond to an absence of a value to avoid memory leaks. 

Besides, each spanning edge removal operation publishes it to a special global \texttt{remove\_op} field (line~\ref{line:remove-op}), so a concurrent edge addition can propose its edge as an replacement to restructure the connectivity component {---} we use \texttt{SpanningEdgeRemoveOp} descriptor for this purpose.
Finally, all non-spanning edges are stored in a multi-set \texttt{non\_spanning\_edges} (line~\ref{line:multiset}). For fine-grained locking over components of connectivity, \texttt{remove\_op} field and the multi-set should be stored in nodes of the Cartesian trees instead.

\paragraph{The \removeEdge{} Operation}
The edge removal procedure (lines~\ref{line:remove-edge-start}--\ref{line:remove-edge-finish}) is straightforward. When the status is already \texttt{REMOVED}, then the operation can complete immediately {---} the edge is already removed (line~\ref{line:removed-return}). Similarly, when the status is \texttt{INITIAL}, the edge is not added yet, so we can linearize this removal before its addition and finish (line~\ref{line:nonspanning:initial-return}). In the case of the \texttt{SPANNING} status, the corresponding \ett{} should be modified, so the operation takes the lock and performs this removal as ``writer'' (line~\ref{line:blocking-remove-inv}). In the last case, the status is \texttt{NON\_SPANNING}. Thus, changing it atomically to \texttt{REMOVED} \emph{logically} removes the edge (line~\ref{line:remove-logical}). After that, we delete it \emph{physically} from \texttt{non\_spanning\_edges} (line~\ref{line:remove-physical}).

The replacement search, while iterating over the multi-set of non-spanning edges, can see some of them in the \texttt{INITIAL} status. For correctness purposes, it should help to add such edges using an algorithm almost identical to the non-spanning edge addition described later. 

\paragraph{The \addEdge{} Operation}
The edge addition procedure starts with a connectivity check (line~\ref{line:add-connectivity-check}) and falls to the blocking mode when the edge connects different components (line~\ref{line:add-blocking-inv}).
Then, the algorithm optimistically adds the edge to \texttt{non\_spanning\_edges} (line~\ref{line:add-optimistical}) {---} it will later remove it from the multi-set if the edge cannot be non-spanning.

Once the edge is added to \texttt{non\_spanning\_edges}, the operation may notice that there is a concurrent spanning edge removal for which this adding edge can be a replacement to reconstruct the Euler Tour Tree (lines~\ref{line:add-notice-start}--\ref{line:add-notice-end}). In this case, we try to set our edge to the \texttt{replacement} field of the removal operation (line~\ref{line:add-try-set-repl}) {---} remember, the edge status is still \texttt{INITIAL}. If the edge is installed to \texttt{replacement}, the operation removes the edge from the multi-set (line~\ref{line:add-multiset-remove}), changes the status to \texttt{SPANNING} (line~\ref{line:add-status-update}) and completes. Note that the status update failure (line~\ref{line:add-status-update}) can be caused only if the concurrent edge removal already processed the edge as replacement and made it spanning (or even removed after that).

When the edge could be a replacement but the remove operation is logically completed (line~\ref{line:add-remove-complete}), then the edge will connect different components. Therefore, we delete it from the multi-set of non-spanning edges and fall back to the blocking mode (lines~\ref{line:add-remove-complete-start}--\ref{line:add-remove-complete-end}).

We successfully examined the case when a concurrent removal operation has been detected. Since the edge is already added to \texttt{non\_spanning\_edges}, if a concurrent removal starts after we read \texttt{null} from \texttt{remove\_op}, it should find our edge. However, we still need to re-check that the edge does not connect different components, removing the edge from \texttt{non\_spanning\_edges} and re-starting the operation in this case (line~\ref{line:add-recheck-start}--\ref{line:add-recheck-end}).

Finally, we can add the edge as a non-spanning by an atomic status update from \texttt{INITIAL} to \texttt{NON\_SPANNING} (line~\ref{line:add-final-update}). If the corresponding \texttt{CAS} succeeds, the edge is successfully added as a non-spanning, so the operation completes. Otherwise, if the \texttt{CAS} fails, a concurrent ``writer'' either has used the edge as a replacement to reconstruct the \ett{} or has helped to add the edge, so one copy of it should be removed from the multi-set of non-spanning edges, and the operation completes (lines~\ref{line:add-failure-start}--\ref{line:add-failure-end}).

\paragraph{Correctness of Non-Spanning Edge Addition}
To show the correctness of the described non-spanning edge modification, we prove the lemma below, which shows the correctness of non-spanning edge addition with a concurrent edge removal. We consider the operations under locks and isolated non-spanning edge modifications trivial, so the corresponded proofs are omitted.

\begin{theorem}
In the proposed non-blocking non-spanning edge addition algorithm (Listing~\ref{lst:non_spanning_simple}), the conflict with a concurrent spanning edge removal that leads to a non-spanning edge between different components of connectivity is impossible.
\label{theorem:non-spanning-edge-correctness}
\end{theorem}

\begin{proof}
Proof by contradiction. First, if there is a replacement for the removed edge, then the concurrent addition does not play any role. Similarly, if there is no replacement and the edge of the concurrent addition cannot be a replacement, the conflict is impossible. Therefore, we consider the situation when the non-spanning edge addition provides the only potential replacement for the concurrent removal.

If the non-spanning addition inserts the edge to \texttt{non\_spanning\_edges} (line~\ref{line:add-optimistical}) before the removal operation starts its blocking part, then the removal operation will see the edge and help complete its addition, so there will be no conflict. 
On the other side, if the removal operation finishes before the connectivity check in the addition algorithm (line~\ref{line:add-final-update}), this check will fail and cause the operation to fall back to the blocking version, so the conflict will not arise. 
Therefore, the only interesting case is when the removal operation still holds the lock, while the non-blocking addition is between the lines~\ref{line:start}-\ref{line:end}. 
As a consequence, the \texttt{can\_be\_replacement} invocation on the line~\ref{line:can-be-repl} will return correct results if a conflict with a removal is possible. Consider what happens first: the removal operation publishes itself in the \texttt{remove\_op} or the non-spanning addition operation reads \texttt{remove\_op} (line~\ref{line:add-notice-start}).

\emph{Case 1:} the publication is the first, so the addition will see the removal operation and will suggest the edge as a replacement instead of completion of the addition. If the replacement proposal (\texttt{CAS}) fails, it means that there is another edge in the slot or the removal closed the slot; both situations are correct. %In case when the edge was written in the slot, but a concurrent thread changed its status to \texttt{NON-SPANNING}, causing the edge to be eventually removed from the slot, this concurrent addition should have read \texttt{root.removal\_operation} before its publishing, so this situation is the same as the second case.

\emph{Case 2:} the read is the first, consequently, the information of the edge is inserted before the search, because the publication precedes the search and the information insertion precedes the read (line~\ref{line:add-optimistical}). Thus, the search will see the edge and will help to add it if necessary.
\end{proof}

\paragraph{Implementation Details} 
The discussed simplified version can be adapted to the Holm et al. algorithm that we use as a sequential version under the hood. Due to space limit and dramatic complexity increase, we move the pseudo-code of our full algorithm with all the implementation details to Appendix~\ref{appendix:non-spanning}. 
%Our implementation will be made open-source upon publication.  

\subsection{Time Complexity Analysis}
Consider an abstract fully lock-free generalization of the Holm et al. algorithm for the dynamic connectivity problem, e.g., the one implemented via lock-free software transactions. Since the replacement search takes $\Omega(N)$ in the worst case, and other operations can make it fail and re-run, the tight upper bound on the total work is $\Omega(MN)$, which is substantially worse than $\O(M \log^2 N)$ of the sequential algorithm. This observation motivated us to design spanning edge updates blocking rather than lock-free.

Theorem~\ref{theorem:time} below shows the total work performed by all threads \emph{excluding the work spent on waiting for locks} for our full algorithm. The same bound can be used when some of the optimizations (such as non-blocking modifications) are switched off since they mainly reduce the need for locks, which is not essential in the analysis.

\begin{theorem}\label{theorem:time}
 The total amount of work excluding the work spent on locking is $\O(M \log^2 N + P\cdot M\log N)$, where $P$ is the number of processes, $M$ is the total number of operations, and $N$ is the number of graph vertices.
\end{theorem}
\begin{proof}
Obviously, for $P=1$ our algorithm performs the same amount of work as the sequential algorithm by Holm et al. To prove the theorem in the concurrent scenario, we need to show that each successful modification can cause at most $\O(\log N)$ additional work for every other thread.
\begin{itemize}
    \item A successful spanning forest modification induces $\O(\log N)$ extra work for concurrent connectivity queries, non-spanning additions, or non-spanning removals due to re-tries. 
    \item Non-blocking additions and removals cause at most $\O(1)$ work for other modifications due to edge state changes and at most $\O(\log N)$ work for a replacement edge search.
    % due to ``has non-spanning edge`` flags may be incorrectly set to \texttt{true}. 
    \item Another $\O(\log N)$ extra work is performed by a blocking removal to find a replacement in case a concurrent addition has already proposed a replacement.
\end{itemize}
\end{proof}
 
% \begin{corollary}
% At least one thread in the system is guaranteed to make progress.
% \end{corollary}
% It is unusual to see lock-free operations have total time complexity increased only in $(1 + \frac {p} {\log n})$ times, because typical lock-free algorithms, such as Treiber stack and Michael-Scott Queue, have worst-case total work increased in $p$ times.

\section{Evaluation}
\label{section:experiments}
To evaluate our concurrent optimizations and their combinations we perform a range of experiments on  workloads composed of both real and synthetic graphs. We consider scenarios with different ratios of reads to updates, which help us evaluate the non-blocking reads, as well as the applicability of our algorithms to different workloads. In addition, we consider scenarios with only additions or only deletions to compare the performance of these two operations. All scenarios are executed on graphs with different properties, e.g. with different graph densities. Our experimental setup is publicly accessible~\cite{our-experiments}.

\subsection{Evaluation Setup}

\paragraph{Benchmarks}
We use three different workloads in the experiments:
\begin{itemize}
    \item \textbf{The Random Subset Scenario}. Here, the data structure is initially filled with half of the graph chosen randomly. After that, multiple parallel threads execute random operations on randomly chosen graph edges. To maintain the number of edges in the structure approximately constant, the percentage of edge additions is equal to the percentage of edge removals, which means that only the ratio of reads to modifications can change. Aksenov et al.~\cite{Aksenov2018} essentially use the same experimental setup. 
    % We stop the threads when the total number of executed operations exceeds some predefined number to avoid result distortion due to ``amortized'' nature of the data structure.
    
    \item \textbf{Incremental Scenario}. In this case, threads add concurrently the whole graph to the initially empty data structure. 
    This benchmark is inspired by the \emph{incremental dynamic connectivity problem}.
    
    \item \textbf{The Decremental Scenario}. This benchmark is opposite to the previous one: parallel threads remove graph edges from a data structure initialized with the whole graph.
    
\end{itemize}

%\begin{figure*}[h]
%\centering
%\includegraphics[width=\linewidth]{src/Mix.pdf}
%\caption{Experimental results on the \emph{random} scenario: $50\%$ connectivity checks with $50\%$ updates on the left, and $99\%$ connectivity checks with $1\%$ updates on the right. The full set of experiments is presented in Appendix~\ref{appendix:experiments}.}
%\label{scenario:random:mix}
%\vspace{1em}
%\end{figure*}
    
\paragraph{Graphs}
We used both real graphs and random graphs listed in Tables~\ref{table:small-graphs} and \ref{table:large-graphs}. We divide these graphs into relatively small and large. Smaller graphs help to analyze the dependency between the number of threads and the algorithm performance. Large graphs are always processed by the maximum number of threads.
%,
The graph of USA roads was taken from a competition on finding shortest paths~\cite{dimacs:challenge9}. Other real-world graphs are from the SNAP graph repository~\cite{snapnets}. The random graphs were generated by the Erdős-Rényi model, where all edges are present with an equal probability. Large synthetic graphs are from another edition of the DIMACS graph processing competition~\cite{dimacs:challenge10}. We remove loops and multi-edges from the graphs, as they do not affect connectivity.
%At the end, we also have a synthetic graph with 10 random components of connectivity for the reason that graphs generated with the Erdős-Rényi model tend to have an only large component of connectivity, and in this case fine-grained locking does not help.

\begin{table}[H]
\centering\renewcommand\cellalign{lc}
\caption{Small graphs used in our experiments for evaluation all the algorithms on various numbers of threads.}
\begin{tabular}{|l|l|l|l|}
\hline
\thead{Graph}                            & \thead{|V|} & \thead{|E|} & \thead{Description}                                                                                                                                \\ \hline
\makecell{USA roads}                      & 435666            & 521200           & \makecell{\small Colorado State roads}                                                                                                             \\ \hline
\makecell{Twitter}                         & 81306             & 1342296          & \makecell{\small Social circles in Twitter}                                                                                                 \\ \hline
\makecell{Stanford web}                    & 281903            & 1992636          & \makecell{\small Web graph of Stanford.edu }                                                                                     \\ \hline
\makecell{\small Random, $|E|=|V|$}            & 400000            & 400000           & \makecell{Random sparse graph}                                           \\ \hline
\makecell{\small Random, $|E|=2|V|$}           & 300000            & 600000           & \makecell{Random sparse graph}                                           \\ \hline
\makecell{\small Random, $|E|=|V| \log |V|$}   & 100000            & 1600000          & \begin{tabular}[c]{@{}l@{}} \makecell{Random dense graph}\end{tabular}                                     \\ \hline
\makecell{\small Random, $|E|=|V| \sqrt{|V|}$} & 20000             & 1600000          & \begin{tabular}[c]{@{}l@{}} \makecell{Random high-density
graph}\end{tabular}                               \\ \hline
\makecell{Random, 10 components}         & 100000            & 1600000          & \begin{tabular}[c]{@{}l@{}}\makecell{Random graph with 10 components}\end{tabular} \\ \hline
\end{tabular}
\label{table:small-graphs}
\end{table}

\begin{table}[H]
\centering\renewcommand\cellalign{lc}
\caption{Large graphs used in our experiments for evaluating all the algorithms on the maximum (144) threads. As the experiments on large graphs are time-consuming, and the trends seem similar to the ones on smaller graphs (listed in Table~\ref{table:small-graphs}), we find it sufficient to have experiments only on the maximum number of threads.}
\begin{tabular}{|l|l|l|l|}
\hline
\thead{Graph}                            & \thead{|V|} & \thead{|E|} & \thead{Description}                                                                                                                                \\ \hline
\makecell{Full USA roads}                      & 23.9M            & 28.9M           & \makecell{The graph of all roads in the USA}                                                                                                             \\ \hline
\makecell{LiveJournal}                         & 4.8M             & 42.9M          & \makecell{\small Friendship graph in LiveJournal}                                                                                                 \\ \hline
\makecell{Kron}                    & 2.1M            & 91M          & \makecell{\small Synthetic Kronecker graph }                                                                                     \\ \hline
\makecell{ Random}            & 4.2M            & 48M           & \makecell{Erdös-Rényi random graph}                                           \\ \hline
\end{tabular}
\label{table:large-graphs}
\end{table}

\paragraph{Hardware and Software}
All the algorithms presented in this paper are implemented in Kotlin for JVM.
The experiments were run on a server with 4 Intel Xeon Gold 6150 (Skylake) sockets; each socket has 18 2.70 GHz cores, each of which multiplexes 2 hardware threads, for a total of 144 hardware threads. We used OpenJDK 11\footnote{We use Java 11 because \texttt{jmh-gradle-plugin} does not work with the recent Java versions, but this should not affect the analysis significantly.} in all the experiments and the Java Microbenchmark Harness (JMH) for running benchmarks~\cite{jmh}.

\subsection{Algorithm Combinations}

Our experiments examine a wide variety of techniques and scenarios. We enumerate them to make our plots easier to read.
First, (1) coarse-grained locking can be used to synchronize \emph{all} operations. We can improve it by (2) using a readers-writer lock instead, (3) using the non-blocking \connected{} implementation, leveraging hardware transactions (HTM) via lock elision~\cite{htm, rajwar2001speculative}. This can be done for (4) all operations, or (5) only for updates.
Similarly, (6) the solution with fine-grained locks can be improved with (7) readers-writer locks and with (8) non-blocking \connected{}. In addition, the algorithm for non-blocking non-spanning edge updates can be applied to the versions with both (9) fine-grained and (10) coarse-grained locking, each uses non-blocking \connected{} in our experiments. 
It is also possible to use non-blocking non-spanning edge updates and \connected{}-s with HTM (11) for spanning edge additions and removals.
In addition, our optimizations are compared with (12) the \emph{parallel combining} technique~\cite{Aksenov2018}, which processes read operations in parallel, and (13) a similar \emph{flat combining} algorithm, but leveraging our non-blocking reads. We examine the performance of all these variants in the following, keeping the numbering consistent.

\paragraph{Sampling}
Iyer et al. discovered that for better performance, dynamic connectivity algorithms should use \emph{random sampling} even in the deterministic Holm et al. algorithm~\cite{seq-experiments}. This random sampling is employed as a fast path in the replacement search to prevent the expensive promotion of spanning edges to the next level. We see the same tendency in the concurrent scenario, so we use the sampling heuristic for all algorithms in our experiments. Furthermore, it is even more important in the concurrent case as it reduces the average time when spanning edge removals hold the locks.

\begin{figure*}[h!]
\centering
\includegraphics[width=\textwidth]{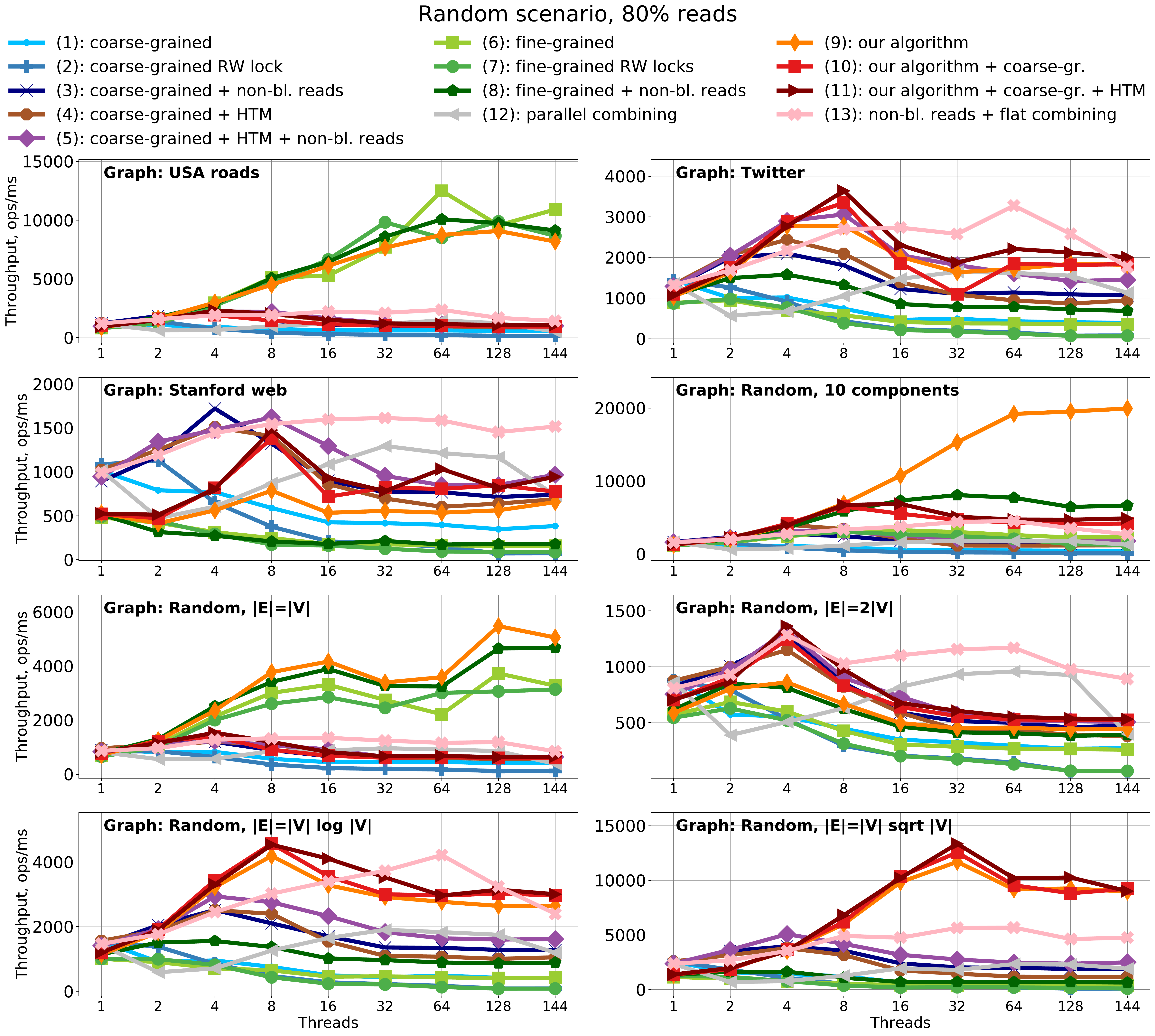}
\includegraphics[width=\textwidth]{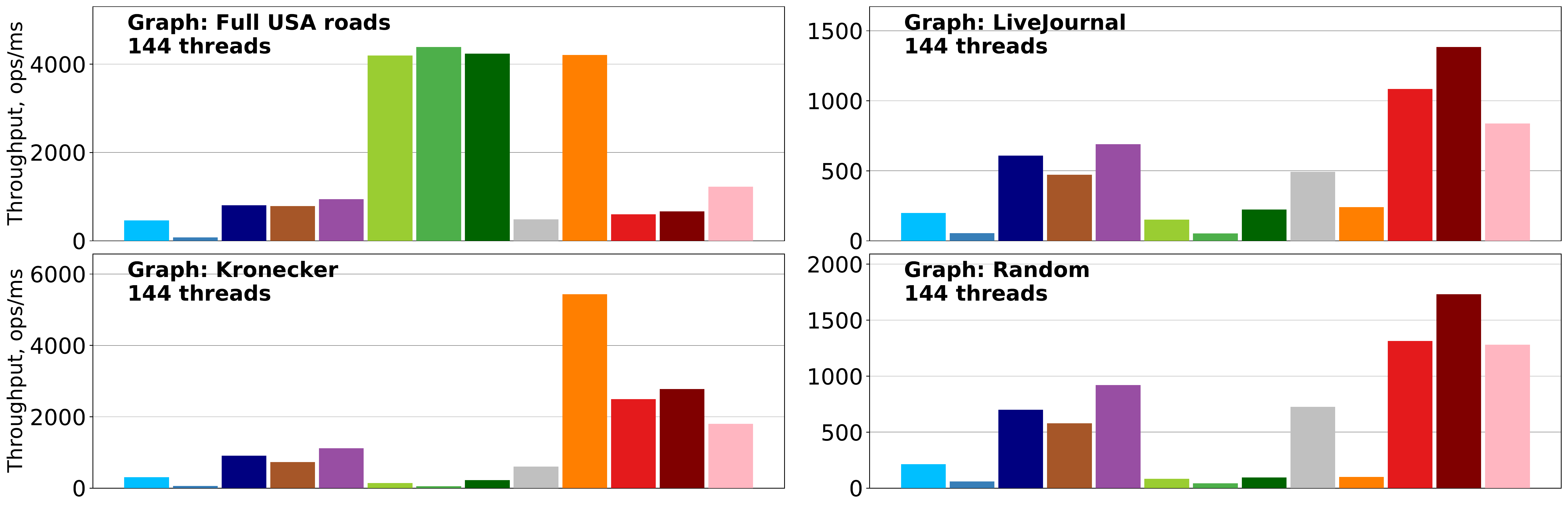}
\caption{Experimental results on the \emph{random} scenario with $80\%$ connectivity checks, $10\%$ edge additions and $10\%$ edge removals. \textbf{Higher is better.} The top plots show the average operation throughput on small graphs, listed in Table~\ref{table:small-graphs}, while the bar charts below show the results on large graphs, listed in Table~\ref{table:large-graphs}, at maximum parallelism.}
\label{scenario:random:80}
\end{figure*}

\begin{figure*}[h!]
\centering
\includegraphics[width=\textwidth]{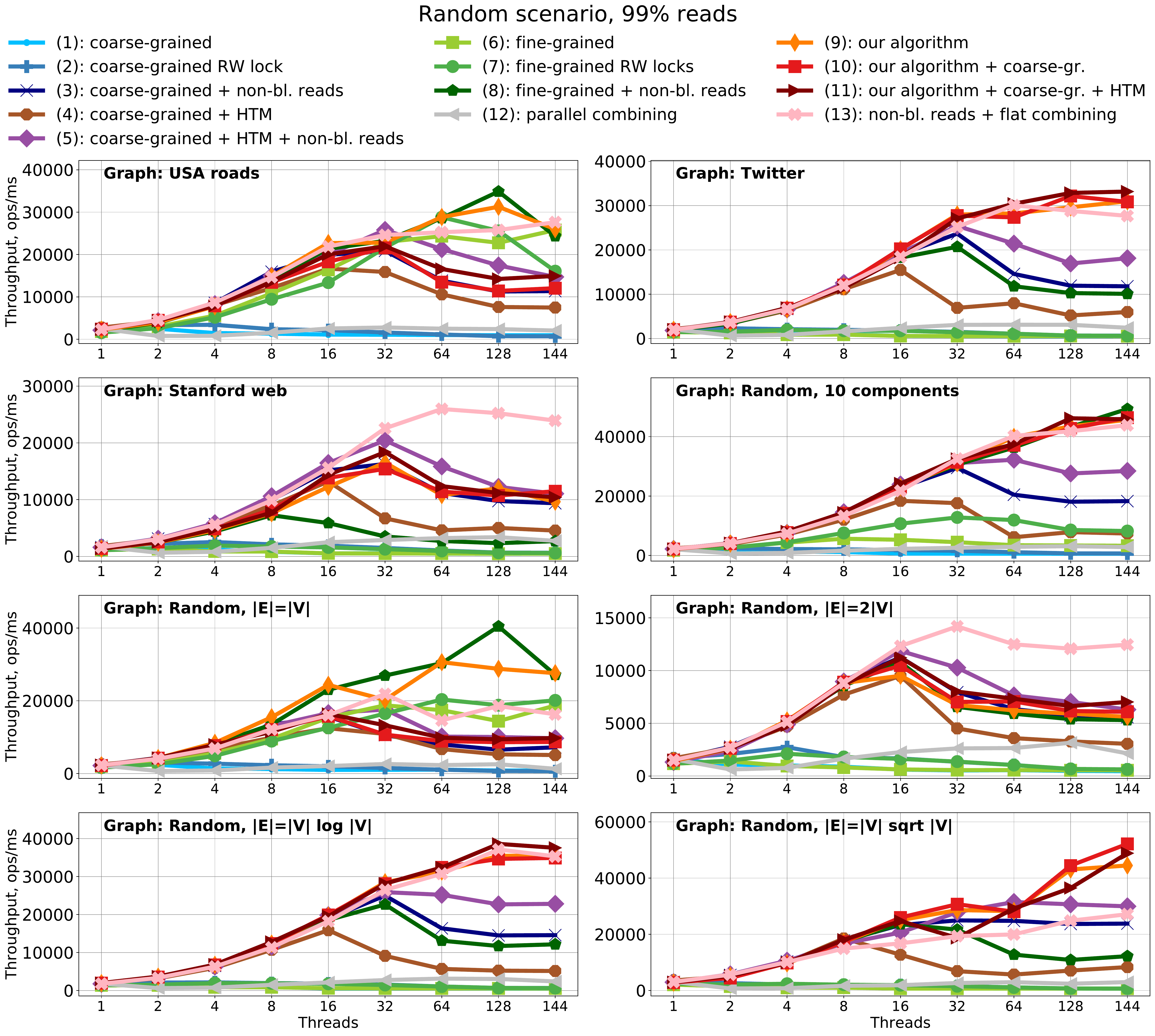}
\includegraphics[width=\textwidth]{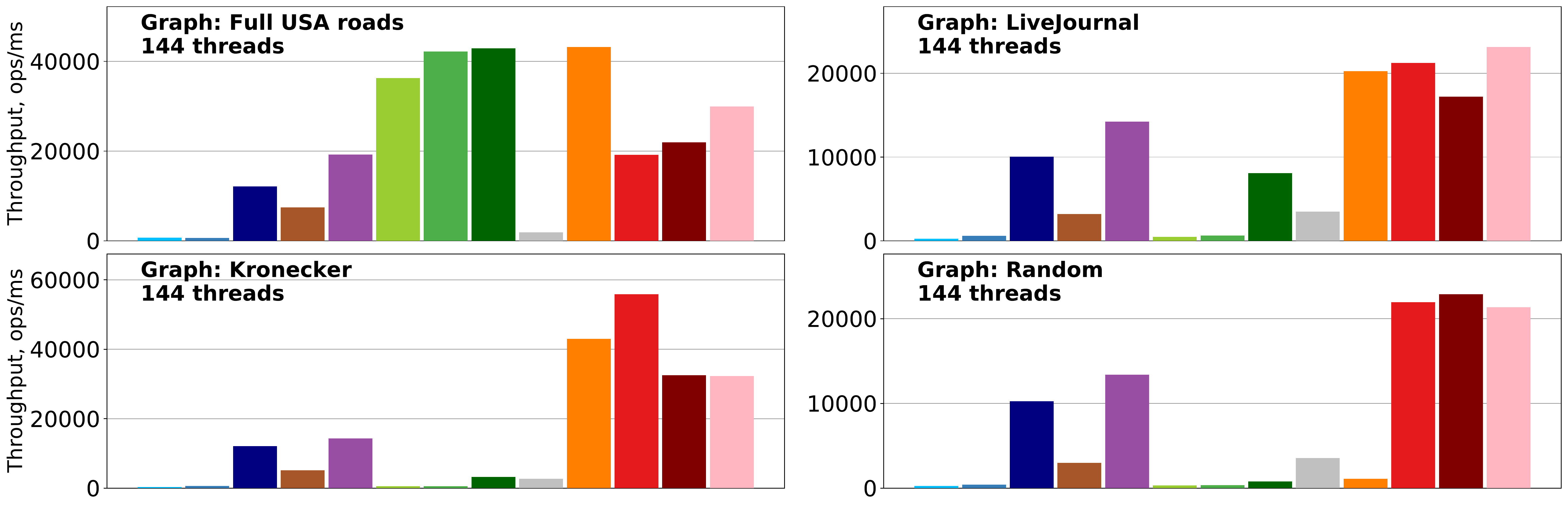}
\caption{Experimental results on the \emph{random} scenario with $80\%$ connectivity checks and $1\%$ updates. \textbf{Higher is better.}  The top plots show the average operation throughput on small graphs, listed in Table~\ref{table:small-graphs}, while the bar charts below show the results on large graphs, listed in Table~\ref{table:large-graphs}, on maximum parallelism.}
\label{scenario:random:99}
\end{figure*}

\subsection{Experimental Results}

\paragraph{Lock-Free Reads} 
The first observation we make is that our non-blocking connectivity query is consistently the most efficient way to implement the \texttt{connected} operation. The readers-writer lock does not scale even for read-dominated scenarios and, in fact, is worse than normal locks. Flat combining with non-blocking reads performs better than the read-optimized flat parallelization technique. Moreover, our non-blocking reads do not have worse single-threaded performance, since the additional memory locations they read are likely to be in the cache. What is more, in fact, more than $99.99\%$ reads succeed on the first try, so there is no thread starving and they are \emph{practically} wait-free.

\paragraph{Fine-Grained Locking} 
The benefits of using fine-grained locking over connected components are less clear-cut. For graphs that have multiple connected components, the algorithm scales substantially better, but for graphs with a single component, it is slower than coarse-grained locking. Surprisingly, there is an exception to this rule. In the incremental and decremental scenarios for the USA roads graph, which has only one connected component, the algorithms with fine-grained locking scale well. We believe that this is the consequence of the fact that the road graph is planar -- when a small number of random edges are removed from a planar graph, the graph quickly loses its connectivity.

\paragraph{Hardware Transactional Memory}
The algorithms with the lock elision technique have better performance than the ones with coarse-grained locking on our random scenario, except for our full algorithm, where the performances match. However, they are worse in modification-only scenarios, because due to contention hardware transactions are likely to be aborted. For our algorithm there is no difference between coarse-grained locking and the lock elision technique, since the goal of our algorithm is to minimize the number of lock acquisitions.

\paragraph{Non-Blocking Modifications}
The algorithm variant with  non-blocking modifications works well for dense graphs, such as the Twitter graph. It can scale in scenarios with many modifications. However, the trade-off is that this algorithm is slower for single-threaded computations. The only detrimental case is the decremental scenario. The reason for the performance drop lies in the fact that spanning edge removals are far slower than non-spanning ones ($\O(\log^2 N)$ vs $\O(1)$), so not much work can be parallelized.

\paragraph{Flat combining}
In most scenarios, the fastest algorithm is either our algorithm with coarse-grained or fine-grained locking, or a flat combining algorithm with our non-blocking reads. Specifically, our algorithm is scalable in many practical scenarios, where queries dominate and updates are relatively infrequent. However, we also emphasize that in ``bad'' scenarios, where there are many spanning edge modifications and a single connected component, the flat combining algorithm is better due to low synchronization costs and better cache locality.

\paragraph{Non-Spanning Edge Modifications}
% First, we show how the optimization on non-spanning edge modification, discussed in Subsection~\ref{subsec:non_spanning_simple}, works on different graphs.
Table~\ref{table:random_statistics} shows the rates of non-spanning edge additions and removals along with the size of the largest connectivity component obtained during the experiments under the \emph{random scenario}. 
% Table~\ref{table:inc_dec_statistics} shows the rate of non-spanning edge additions and removals in the \emph{incremental} and \emph{decremental} scenarios. 
All the data were obtained on small graphs, listed in Table~\ref{table:small-graphs}, and on the sequential workload. Our internal experiments show that the collected statistics do not change with increasing the number of threads.
Here we see that on dense graphs (such as Twitter, Stanford web, and random graphs with $|V| \log |V|$ and $|V| \sqrt {|V|}$ edges) more than $94\%$ of additions and $74\%$ of removals do not require locks. As for the road and sparse random graphs, the optimization for non-spanning edge modifications does not play a significant role {---} this matches with the experimental results. 

\begin{table}[h]
\caption{The \emph{random scenario} statistics: shows the rates of non-spanning edge additions and removals along with the size of the largest connected component (divided by the total number of nodes) discovered during the experiment. The statistics have been collected on the sequential workload.}
\centering
\label{table:random_statistics}
\begin{tabular}{|l|l|l|l|l|}
\hline
\thead{Graph}                             & \thead{\% of non-span. additions} & \thead{\% of non-span. removals} & \thead{Largest component, \%} \\ \hline
USA roads & 6.3 & 1.5 & 0.4                                                                  \\ \hline
Twitter                        & 98.7 & 88.6 & 94.2                                                                  \\ \hline
Stanford web                   & 94.4 & 74.5 & 74.0
                                                               \\ \hline
Random, $|E|=|V|$             & 0.1 & 0.0 & 0.9
                                                                 \\ \hline
Random, $|E|=2|V|$            & 63.4 & 16.0 & 77.8
                                                                \\ \hline
Random, $|E|=|V| \log |V|$    & 100.0 & 87.5 & 100.0
                                                                 \\ \hline
Random, $|E|=|V| \sqrt {|V|}$ & 100.0 & 97.5 & 100.0
                                                                    \\ \hline
Random, 10 components          & 100.0 & 87.4 & 10.0                                                                     \\ \hline
\end{tabular}
\end{table}

\paragraph{Lock Waiting Expenses} 
To show how much time various algorithms wait for the required locks, we collect the statistics on the \emph{active} time rate, which excludes the waiting part. In particular, when the rate is $100\%$, the algorithm spends all the time in the graph processing and no time on waiting for locks. The corresponding data is presented in Figures~\ref{scenario:stats:random:80}--\ref{scenario:stats:random:99}. 
First, fine-grained locking significantly improves the straightforward coarse-grained implementation on sparse graphs but shows similar performance on the dense ones. 
At the same time, applying non-blocking reads always reduces the time spent on waiting for locks {---} this is expected since we use read-dominated scenarios in our experiments. 
As for the non-spanning edge updates optimization, the impact correlates with the already discussed statistics presented in Table~\ref{table:random_statistics} {---} the lock waiting time is significantly reduced on dense graphs.

\begin{figure}[h]
\centering
\includegraphics[width=0.8\textwidth]{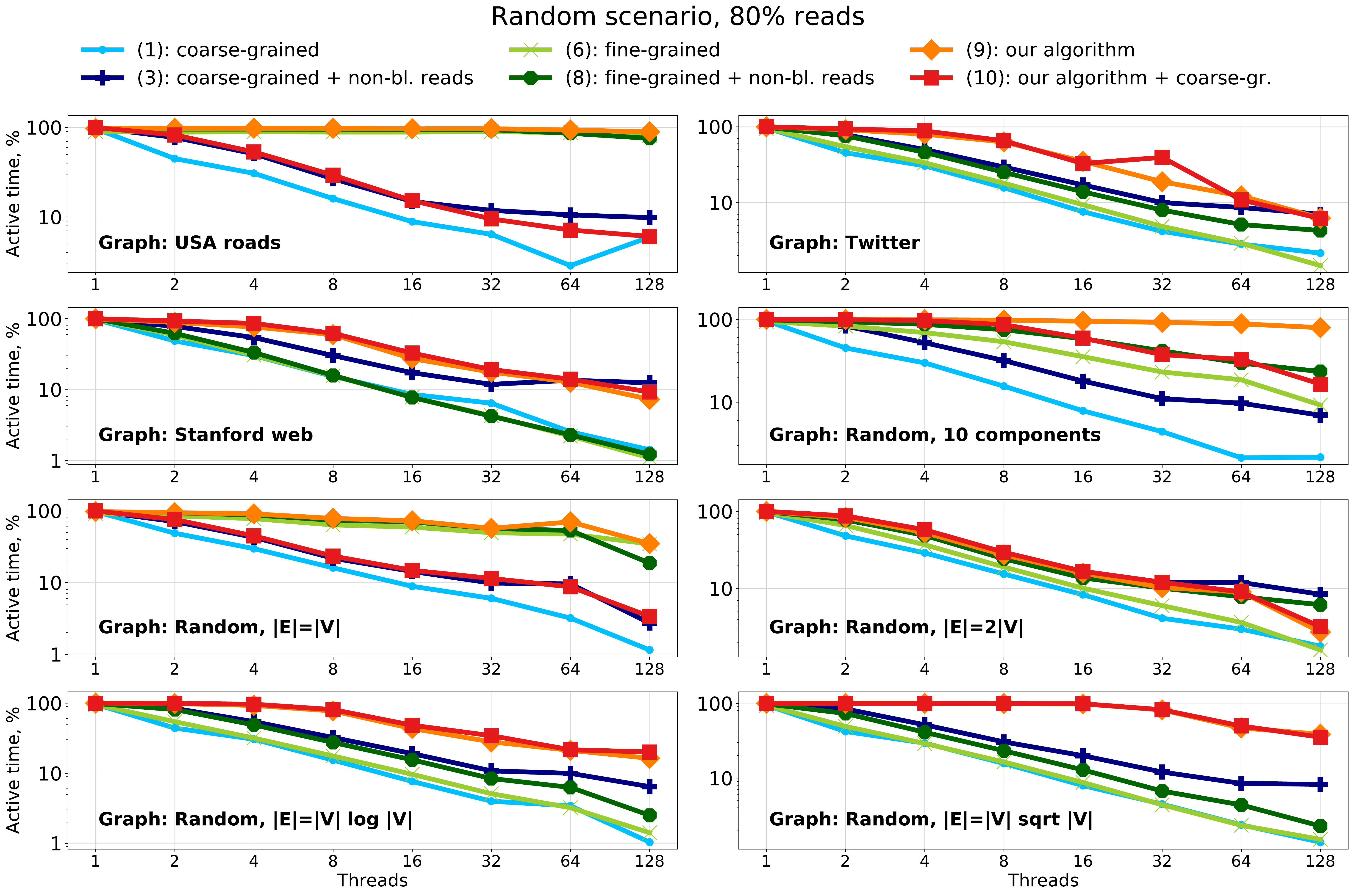}
\caption{The \emph{active time} rate for the \emph{random} scenario with $80\%$ connectivity checks and $20\%$ updates that shows the graph processing time and excludes the time spent on waiting for locks. \textbf{Higher is better, 100\% is the best possible.}}
\label{scenario:stats:random:80}
\end{figure}

\begin{figure}[h]
\centering
\includegraphics[width=0.8\textwidth]{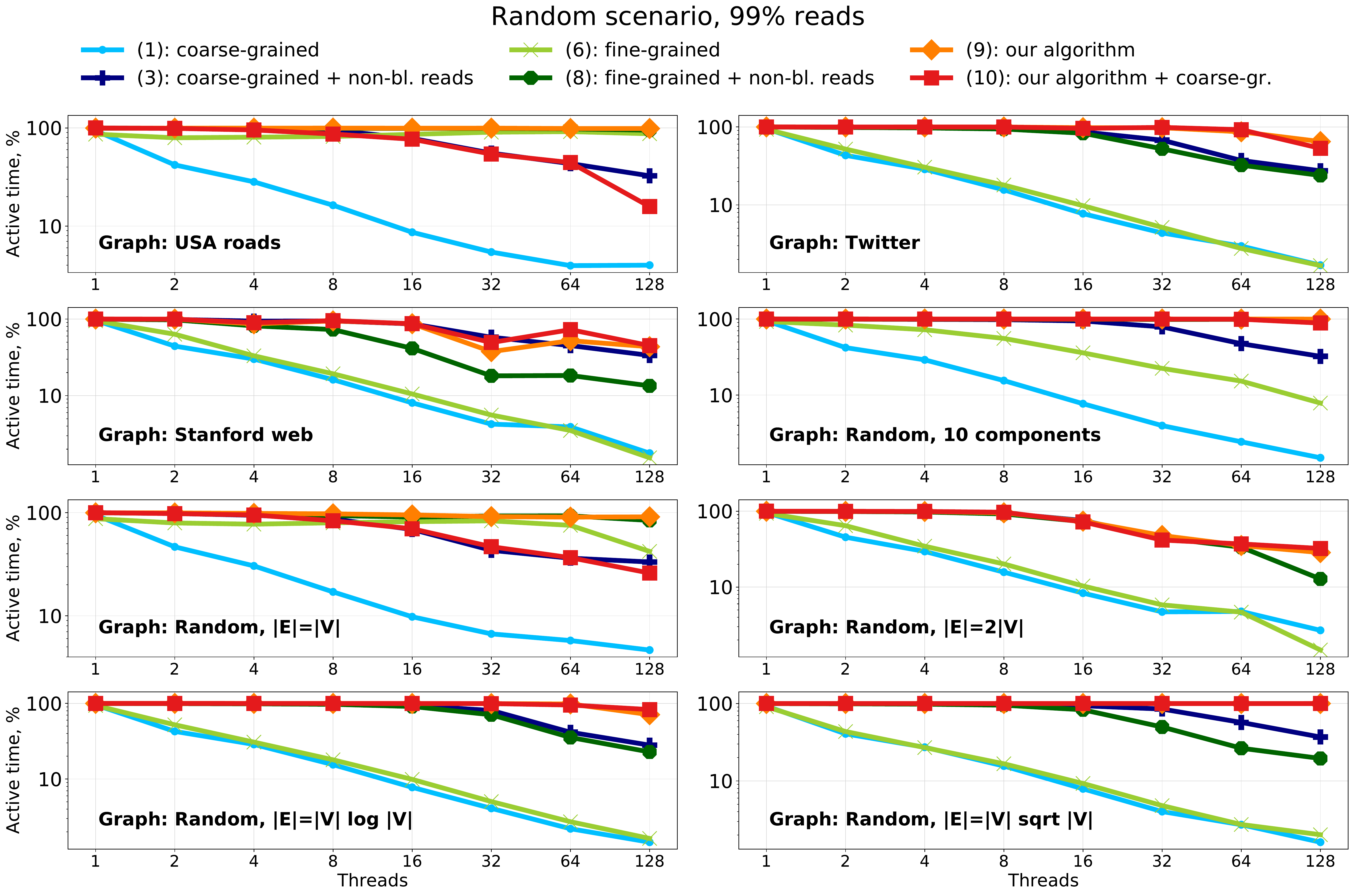}
\caption{The \emph{active time} rate for the \emph{random} scenario with $99\%$ connectivity checks and $1\%$ updates that shows the graph processing time and excludes the time spent on waiting for locks. \textbf{Higher is better, 100\% is the best possible.}}
\label{scenario:stats:random:99}
\end{figure}

\paragraph{Incremental and Decremental Scenarios.} We omit the experimental data for the \emph{incremental} and \emph{decremental} scenarios in the main body and present them in Appendix~\ref{appendix:experiments}.

\section{Conclusion}

We presented a first concurrent variant of the dynamic connectivity, based on the classic sequential algorithm of Holm et al.~\cite{Holm2001} and a concurrent single-writer Euler Tour Tree data structure. We have investigated a wide range of technical options for implementing our concurrent design, with various trade-offs. While we are not able to identify a single ``best-performing'' algorithm across all scenarios we investigated, our experimental validation does identify several rules-of-thumb, which should be useful to guide practical implementations. 
In terms of future work, we plan to investigate NUMA-aware and recoverable versions of these algorithms.

\bibliographystyle{unsrt} 
\bibliography{source}

\appendix

%\onecolumn
\clearpage

%\twocolumn
\section{Non-Linearizability Problem when Omitting the Last Root Search in Listing~\ref{nonblocking-reads}}\label{appendix:ett_counterexample}

\begin{theorem}
Omitting the last root search in the algorithm from Listing~\ref{nonblocking-reads} (line~\ref{line:connected:last_search}) results in a non-linearizable connectivity check.
\end{theorem}

For readability, we repeat the \connected{} algorithm for single-writer Euler Tour Trees. 

\begin{figureAsListing}

\end{figureAsListing}

\begin{proof} 

We show a counter-example for the algorithm without the last root search. 

Consider four nodes: $u$, $v$, $w$, $r$. Assume that they are connected with the following parent links: $u \rightarrow w$, $v \rightarrow w$, $w \rightarrow r$. 
Figure~\ref{counter-example} illustrates the structure.

\begin{figure}[h]
\centering
\includegraphics[width=0.2\textwidth]{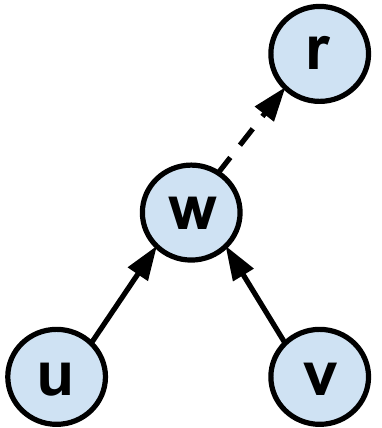}
\caption{Illustration for the counter-example to the connectivity check without the last root search.}
\label{counter-example}
\end{figure}

A reader checks connectivity of $u$ and $v$. Initially, all versions are $0$. The following interleaving with one reader (\textbf{R}) and one writer (\textbf{W}) leads to a non-linearizable result:
\begin{enumerate}
    \item \textbf{R:} The reader's first root search starts and reaches vertex $r$.
    \item \textbf{W:} A concurrent writer removes the edge $(w, r)$; thus, removing the parent link $w \rightarrow r$. The versions of $r$ and $w$ are increased by $1$. 
    \item \textbf{W:} After that, the writer starts an addition of the previously removed edge, but updates only versions of $r$ and $w$ at the current moment. Therefore, the versions of $r$ and $w$ become $2$.
    \item \textbf{R:} The first root search finishes with the result $(r, 2)$. 
    \item \textbf{R:} The second search starts, reaches $w$, finds out that it has no parent, and becomes preempted by the scheduler.
    \item \textbf{W:} The writer finishes the addition of edge $(w, r)$, adding parent link $r \rightarrow w$. 
    \item \textbf{W:} The writer starts removing edge $(w, r)$ again. It increments the version of $w$ at first.
    \item \textbf{R:} The reader reads the version of $w$ equal to $3$ and finishes the root search. 
    \item \textbf{R:} The reader starts the third root search and finds the root $r$ with version $2$.
    \item \textbf{W:} The writer finishes the removal.
    \item \textbf{R:} The fourth root search finds root $w$ with version $3$.
\end{enumerate}

The vertices $u$ and $v$ never were in different components. However, the \connected{} algorithm without the last root search returns \texttt{false}.

\end{proof}

\newpage
%\onecolumn
\section{Evaluation on the Incremental and Decremental Scenarios}\label{appendix:experiments}
% In the main part of the paper we discussed several benchmarks, which results are omitted due to space limit. This appendix section provides the full set of results obtained during the experiments. 

% \vspace{1em}

\begin{figure}[h!]
\centering
\includegraphics[width=0.95\textwidth]{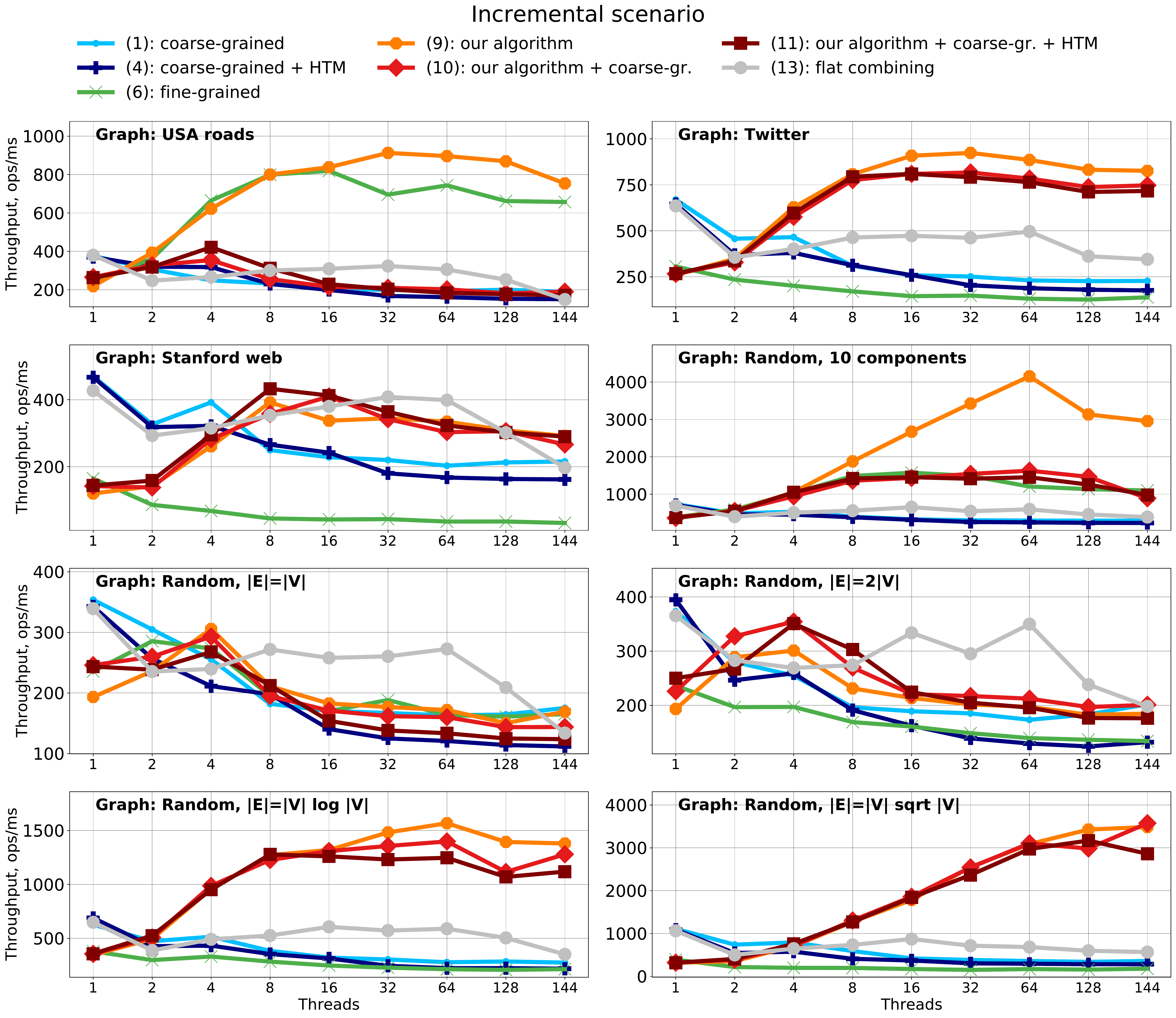}
\includegraphics[width=0.95\textwidth]{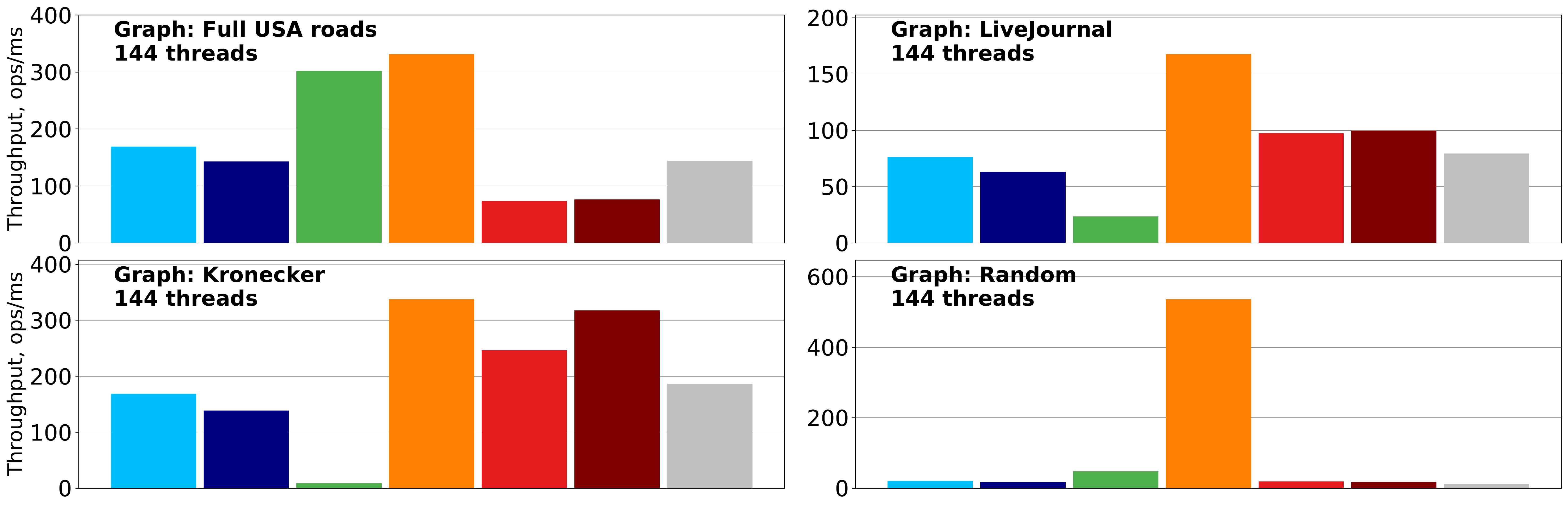}
\caption{Experimental results on the \emph{incremental} scenario. \textbf{Higher is better.} The top plots show the average operation throughput on small graphs, listed in Table~\ref{table:small-graphs}, while the bar charts below show the results on large graphs, listed in Table~\ref{table:large-graphs}, on the maximum number of threads.}
\label{scenario:incremental}
\end{figure}

\begin{figure}[h!]
\centering
\includegraphics[width=0.95\textwidth]{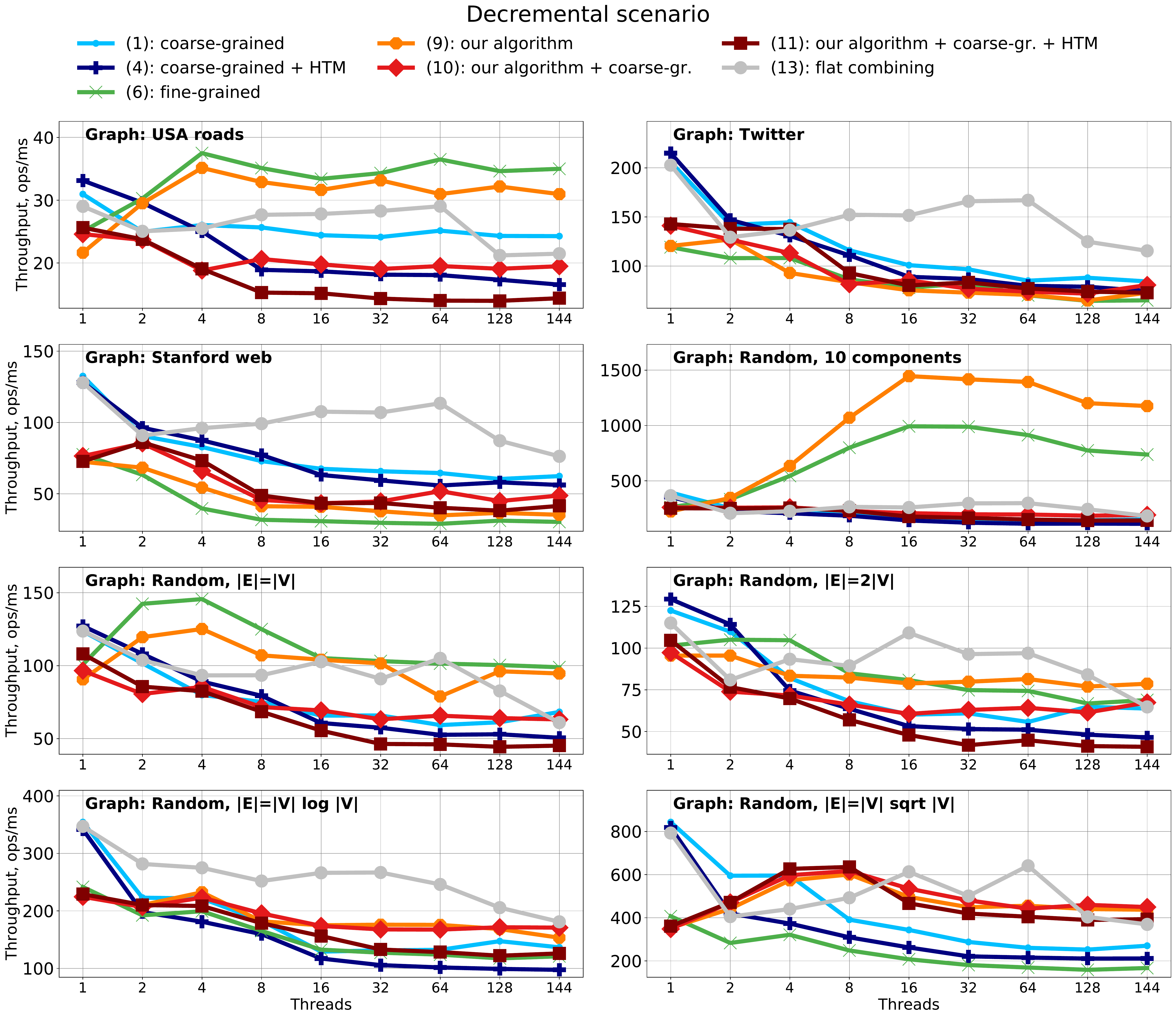}
\includegraphics[width=0.95\textwidth]{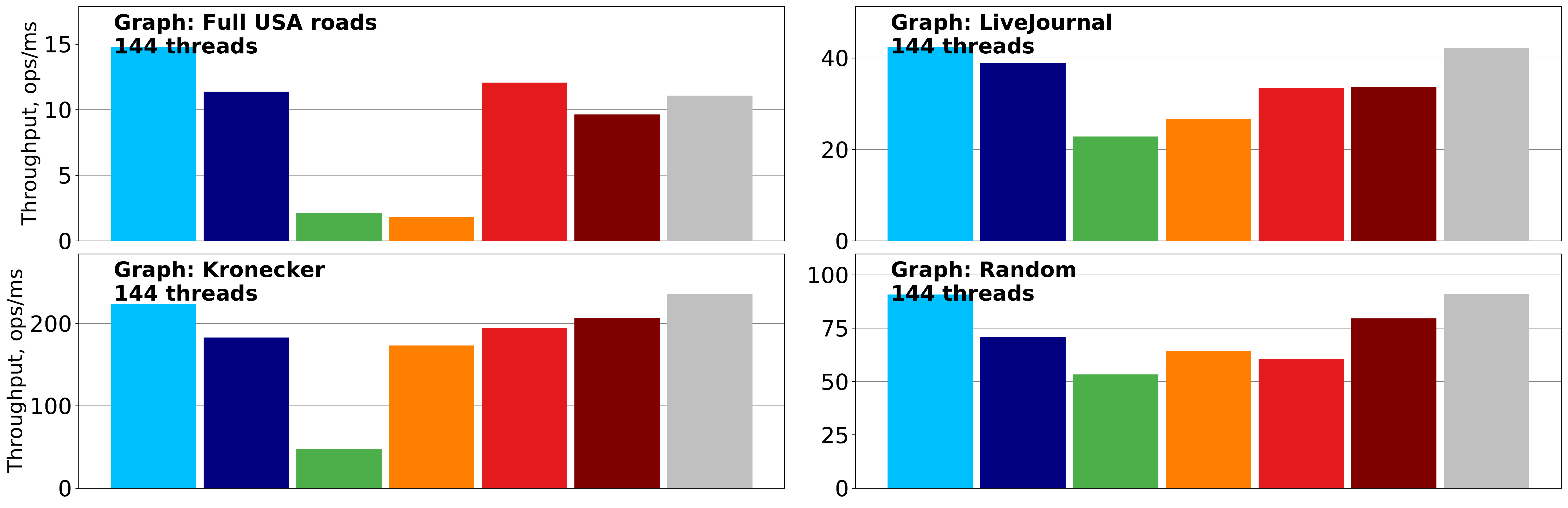}
\caption{Experimental results on the \emph{decremental} scenario. \textbf{Higher is better.} The top plots show the average operation throughput on small graphs, listed in Table~\ref{table:small-graphs}, while the bar charts below show the results on large graphs, listed in Table~\ref{table:large-graphs}, on the maximum number of threads.}
\label{scenario:decremental}
\end{figure}

\newpage

% \begin{figure*}[h]
% \centering
% \includegraphics[width=\textwidth]{src/Large random scenario, 80 reads.pdf}
% \caption{Experimental results on the \emph{random} scenario with $80\%$ connectivity checks and $20\%$ updates for large graphs.}
% \label{scenario:large:random:80}
% \end{figure*}

% \begin{figure*}[h]
% \centering
% \includegraphics[width=\textwidth]{src/Large random scenario, 99 reads.pdf}
% \caption{Experimental results on the \emph{random} scenario with $99\%$ connectivity checks and $1\%$ updates for large graphs.}
% \label{scenario:large:random:99}
% \end{figure*}

% \begin{figure*}[h]
% \centering
% \includegraphics[width=\textwidth]{src/Large incremental scenario.pdf}
% \caption{Experimental results on the \emph{incremental} scenario for large graphs.}
% \label{scenario:large:ncremental}
% \end{figure*}

% \begin{figure*}[h]
% \centering
% \includegraphics[width=\textwidth]{src/Large decremental scenario.pdf}
% \caption{Experimental results on the \emph{decremental} scenario for large graphs.}
% \label{scenario:large:decremental}
% \end{figure*}

%\twocolumn
% \clearpage
% \section{Addittional Evaluation Statistics}

% In this appendix section we provide some statistics obtained during the experiments.

\newpage

\begin{figure}[h]
\centering
\includegraphics[width=0.9\textwidth]{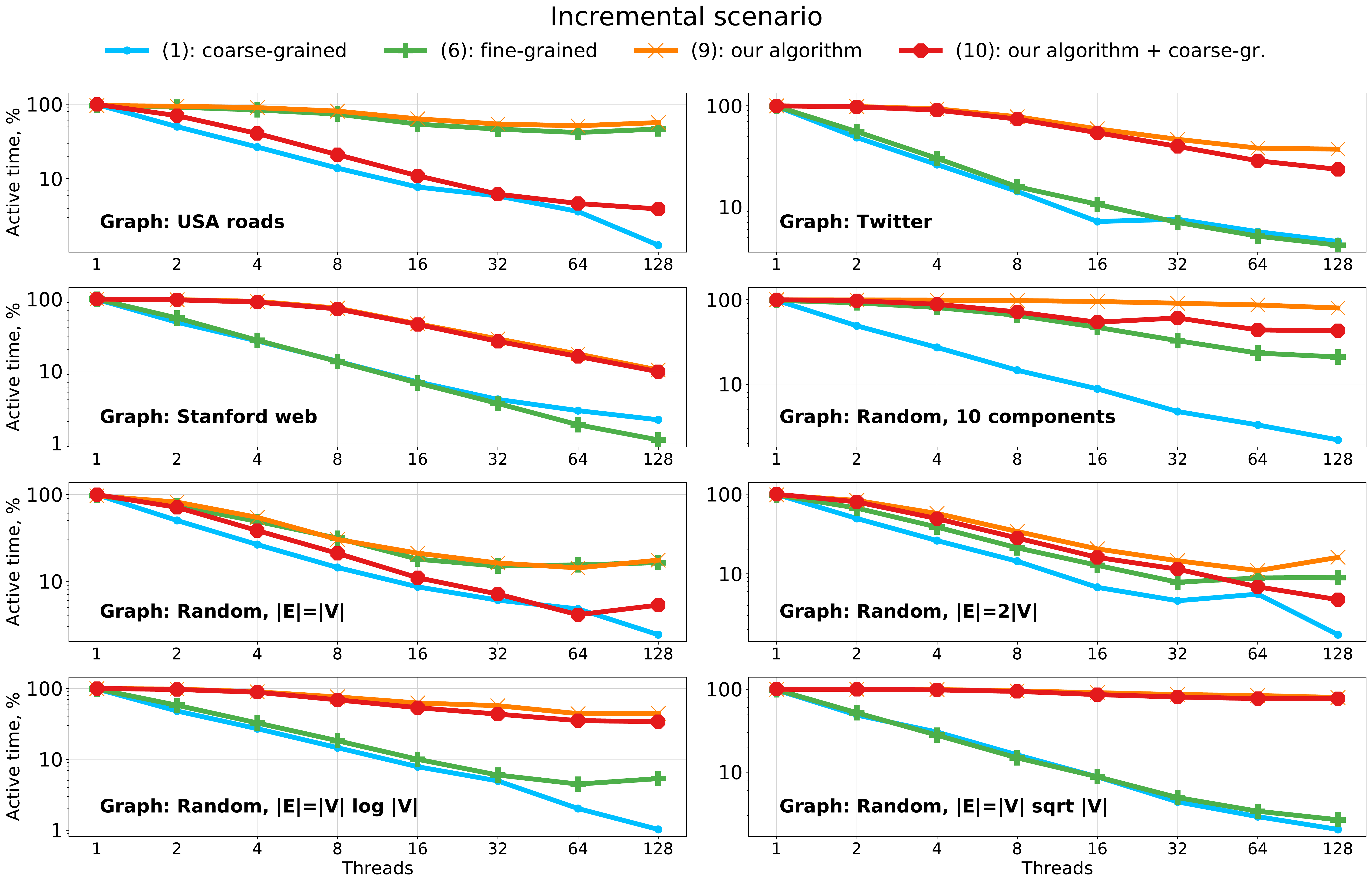}
\caption{The \emph{active time} rate for the \emph{incremental} scenario that shows the graph processing time and excludes the time spent on waiting for locks. \textbf{Higher is better, 100\% is the best possible.}}
\label{scenario:stats:incremental}
\end{figure}

\begin{figure}[h]
\centering
\includegraphics[width=0.9\textwidth]{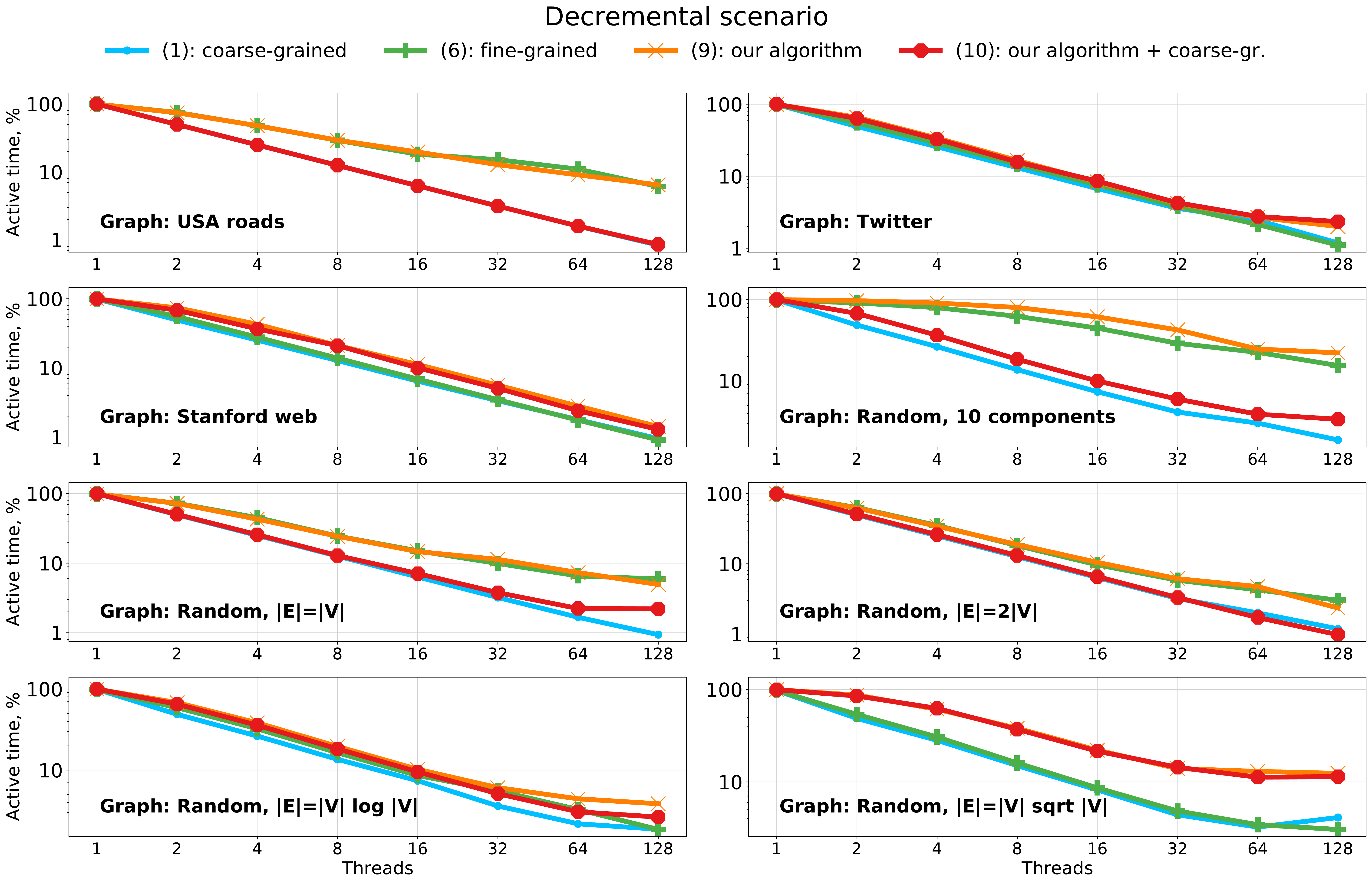}
\caption{The \emph{active time} rate for the \emph{decremental} scenario that shows the graph processing time and excludes the time spent on waiting for locks. \textbf{Higher is better, 100\% is the best possible.}}
\label{scenario:stats:decremental}
\end{figure}

\newpage

\begin{table}[htb]
\caption{The statistics for \emph{incremental} (on the left) and \emph{decremental} (on the right) scenarios, which shows the rates of non-spanning edge additions/removals. The statistics have been collected on the sequential workload.}
\label{table:inc_dec_statistics}
    \begin{subtable}{.5\textwidth}
    \centering
    \caption{Incremental scenario}
            \begin{tabular*}{0.95\linewidth}{@{\extracolsep{\fill}}|l|l|}
\hline 
\thead{Graph}                             & \thead{$\%$ of non-spanning\\ additions} \\ \hline
USA roads                       & 3.7                          \\ \hline
Twitter                          & 81.3                         \\ \hline
Stanford web                     & 66.8                         \\ \hline
Random, $|E|=|V|$             & 2.9                          \\ \hline
Random, $|E|=2|V|$            & 19.4                         \\ \hline
Random, $|E|=|V| \log |V|$    & 76.9                         \\ \hline
Random, $|E|=|V| \sqrt {|V|}$ & 93.4                         \\ \hline
Random, 10 components          & 76.9                         \\ \hline
        \end{tabular*}
        
 \end{subtable}%
%  \hfill
   \begin{subtable}{.5\textwidth}
   \centering
   \caption{Decremental scenario}
           \begin{tabular*}{0.95\linewidth}{@{\extracolsep{\fill}}|l|l|}
\hline 
\thead{Graph}                             & \thead{$\%$ of non-spanning\\removals} \\ \hline
USA roads                       & 3.7                          \\ \hline
Twitter                          & 81.3                         \\ \hline
Stanford web                     & 66.8                         \\ \hline
Random, $|E|=|V|$             & 2.9                          \\ \hline
Random, $|E|=2|V|$            & 19.4                         \\ \hline
Random, $|E|=|V| \log |V|$    & 76.9                         \\ \hline
Random, $|E|=|V| \sqrt {|V|}$ & 93.4                         \\ \hline
Random, 10 components          & 76.9                         \\ \hline
        \end{tabular*}
             
    \end{subtable}
\end{table}

%\newpage
\clearpage
%\twocolumn

\section{Lock-Free Non-Spanning Edge Updates: Full Algorithm and Implementation Details}\label{appendix:non-spanning}

In Subsection~\ref{subsec:non_spanning_simple} of the main body, we presented lock-free non-spanning edge updates in the simplified case when concurrent additions of the same edge are prohibited,  there is a \emph{global} set of non-spanning edges, and the algorithm uses coarse-grained locking. In this appendix section we show how to get rid of these simplifications.

\paragraph{Edge Statuses}
The main difference in the state machine is that we need an additional status \texttt{IN PROGRESS} for spanning edge additions. Previously, we changed the edge status at the beginning of spanning edge addition to indicate concurrent removals that they should take locks to proceed. 
However, concurrent additions of the same edge can not distinguish the following two cases: the case when the edge is already presented and nothing should be done, and the case when the edge is being added concurrently and the operation should wait until its addition is finalized. To fix the problem, at the beginning of spanning edge addition the edge status becomes \texttt{IN PROGRESS} and is changed to \texttt{SPANNING} only after all necessary ETT modifications. The updated state machine is presented in Figure~\ref{status-diagram-complete}.

\begin{figure}[h]
\centering
\includegraphics[width=0.5\linewidth]{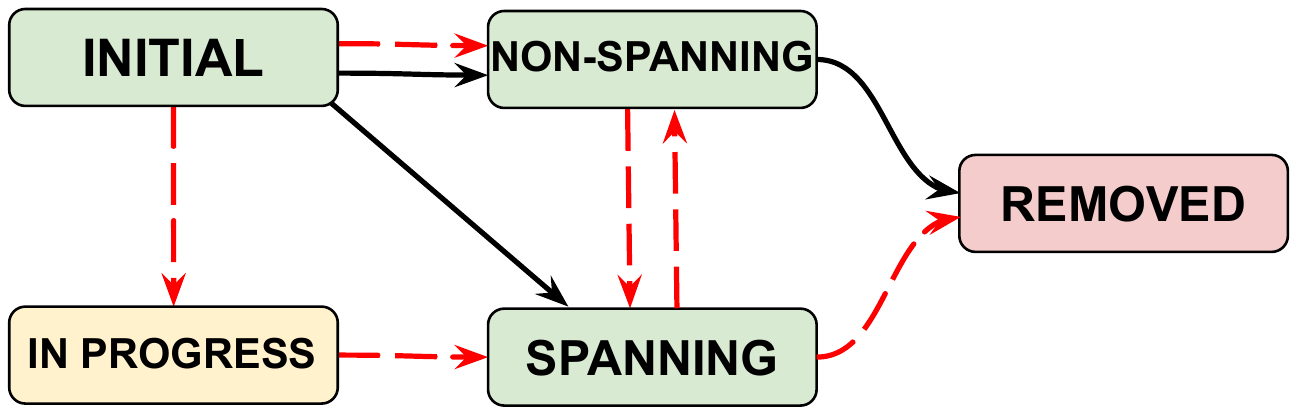}
\caption{Possible transitions of edge statuses in the full version of the algorithm, without assumptions. Red dashed transitions are performed under locks.}
\label{status-diagram-complete}
\end{figure}

\paragraph{The Data Structure}
Listing~\ref{data-structures-complete} shows data structures used in our algorithm and their interfaces. The main class includes ETT levels (line~\ref{line:data-forest}) and a concurrent hash map for edge statuses and levels. The \texttt{REMOVED} status corresponds to an absence in this hash table. An edge status and a level are stored together because it allows to change them atomically, and as a result, simplifies memory management. For performance purposes, an edge level and a status can be merged to fit in a machine word, since they require only $\O(\log \log N)$ and $\O(1)$ bits respectively. In the case of this optimization, to avoid the ABA problem we pair \texttt{INITIAL} status with random bits, this way an edge can not be added twice if other threads can help.

As for ETT Nodes, each node contains, in addition to pointers to other nodes, a version used for non-blocking reads, a multi-set with non-spanning edges adjacent to the corresponding vertex, flags needed for efficient navigation in the tree, a field for announcing spanning edge removals, and a lock. All these fields, except for the version, either were not present in the simplified case or were global for the whole data structure.

\begin{figureAsListing}
\begin{lstlisting}
class ConcurrentDynamicConnectivity(n: Int) {
  // Spanning forests for each graph level
  #\label{line:data-forest}#val F = ConcurrentEulerTourTree[log2(n) + 1]
  // Level and status for each edge
  #\label{line:data-states}#val states = ConcurrentHashMap<Edge, State>()
  func add_edge(Edge)
  func remove_edge(Edge)
  func connected(Vertex, Vertex): Bool
}

class ConcurrentEulerTourTree {
  func add_edge(Edge)
  func remove_edge(Edge)
  func connected(Vertex, Vertex): Bool
  // Returns the node corresponding to a vertex
  func node(Vertex): ETTNode
  func find_root(Vertex): ETTNode
}

class ETTNode {
  // Parent link and children
  var parent, left, right: ETTNode
  // Version of the node (for root only)
  var version: Int
  // For storing non-spanning adjacent edges
  #\label{line:data-non-spanning-edges}#val edges = ConcurrentMultiSet<Edge>()
  // Flags whether there are spanning
  // or non-spanning edges in the subtree
  #\label{line:data-has-non-spanning}#var has_non_spanning_edges: Bool
  #\label{line:data-has-spanning}#var has_spanning_edges: Bool
  // Concurrent removal operation (for root only)
  #\label{line:data-remove-op}#var removal_op: SpanningEdgeRemoveOp?
  #\label{line:data-lock}#val lock = Lock()
}

class RemovalOperation(val u: Vertex, val v: Vertex) {
  var replacement: Pair<Edge, State>? = null
}

class State(val status: Status, val level: Int)
\end{lstlisting}
\caption{Data structures used by our algorithm. Some fields, methods, and classes are omitted, because they are not important for the concurrent case.}
  \label{data-structures-complete}
\end{figureAsListing}

\paragraph{Edge Management}
For storing non-spanning edges in Cartesian tree nodes our algorithm uses a concurrent lock-free multiset, which allows iterating over its elements. It is a multiset rather than a set because we permit multiple copies of the same edge to be present simultaneously. The invariant that we maintain is that if a non-spanning edge of level $r$ is in the graph, then it should have at least one copy in multisets in corresponding adjacent nodes of $F_r$. This invariant is easy to provably satisfy: an operation that tries to add a non-spanning edge inserts information about it \emph{before} the linearization point, a removal operation deletes the information only \emph{after} the linearization point. Similarly, when raising an level $r$ edge, we firstly add the information to level $r+1$ and only after the \texttt{CAS} transition from \texttt{State(NON-SPANNING, r)} to \texttt{State(NON-SPANNING, r+1)} delete the information from the previous level. If some inserted information becomes ``useless'', e.g. if the \texttt{CAS} failed, the same operation deletes it to prevent memory leaks.  

A more complex problem is that the Cartesian trees should have correct flags ``whether there are non-spanning edges in the subtree of a node'' in order to be able to iterate over all current level non-spanning edges, spending only $\O(\log N)$ work per edge. In fact, it will not break the correctness of the algorithm if a small number of nodes has these flags falsely set to true. A replacement search will traverse some odd nodes, but it can potentially fix their flags via backward recursion. So, the idea is to set all flags to \texttt{true} in the path between a node and the root when adding information about a non-spanning edge to the node, but in the same time not to change any flags when removing the information. In other words, flags can be set to \texttt{false} \emph{only} by a replacement search \emph{under the locks}. However, the replacement search can mistakenly set a flag to \texttt{false}, when there is a concurrent addition that has just added information about a new edge and set the flag to \texttt{true}. Our observation is that a simple re-check after setting to \texttt{false} (lines~\ref{line:flags-recheck-start}-\ref{line:flags-recheck-end} in Listing~\ref{listing:edge-information-management}) can help to avoid incorrect flags (Lemma~\ref{flag-management-theorem}). 

\begin{lemma}
\label{flag-management-theorem}
If \texttt{has\_non\_spanning\_edges} flag management is performed according to Listing~\ref{listing:edge-information-management}, where non-blocking additions just set needed flags to \texttt{true} and flag recalculation re-checks if it has written \texttt{false}, a thread holding the lock can not wrongly see \texttt{false} when there is a completely added non-spanning edge in the subtree.
\end{lemma}
\begin{proof}
We prove the lemma by contradiction. Suppose that at some moment of time a situation, when a node has the flag set to \texttt{false} while there is an edge in its subtree, happened. Consider the first such moment. The situation can not be caused by a thread adding a non-spanning edge, because before the linearization point of the addition all flags in the path between a node and the root are set to \texttt{true}. Consequently, it was caused by a replacement search that recalculated the flag and set it to \texttt{false}. However, the recalculation should have re-checked the flag. Since all flag updates are performed from bottom to up the children of the node should have had correct values, so the re-check should have recognized the conflict and set the flag to \texttt{true}, which leads to a contradiction.  
\end{proof}

\begin{figureAsListing}
  \begin{lstlisting}
func ConcurrentEulerTourTree.add_info(edge: Edge) {
  // Put the edge into the multisets and 
  // set all flags in the paths between 
  // the nodes and roots to `true`.
  this.node(edge.u).edges.add(edge)
  set_flags_up(this.node(edge.u))
  this.node(edge.v).edges.add(edge)
  set_flags_up(this.node(edge.v))
}

func ConcurrentEulerTourTree.remove_info(edge: Edge) {
  // Remove a copy of an edge from the multisets
  // without updating any flags.
  this.node(edge.u).edges.remove(edge)
  this.node(edge.v).edges.remove(edge)
}

func set_flags_up(node: ETTNode) {
  // Check if the flag is already enabled
  if node.has_non_spanning_edges: return
  node.has_non_spanning_edges = true
  parent := node.parent
  // If is not a root, go up
  if parent != null: set_flags_up(parent)
}

func recalculate_flags(node: ETTNode) { 
  // Called only under locks.
  // Calculate, whether there is a non-spanning edge
  // in the subtree.
  should_set_flag := should_have_flag_set(node)
  node.has_non_tree_edges = should_set_flag
  #\label{line:flags-recheck-start}#if (!has_non_tree_edges):
  #\indentrule#  // check whether 'true' was overridden by 'false'
  #\indentrule#  should_set_flag = should_have_flag_set(node)
  #\indentrule#  if (should_set_flag): 
  #\label{line:flags-recheck-end}##\indentrule#  #\indentrule#  node.has_non_tree_edges = true
}

func should_have_flag_set(node: ETTNode): Bool = 
  node.edges.size > 0 or 
  node.left.has_non_tree_edges or
  node.right.has_non_tree_edges
  \end{lstlisting}
  \caption{Edge information managing methods. \texttt{edges} is a concurrent multiset.}
  \label{listing:edge-information-management}
\end{figureAsListing}

\paragraph{The \removeEdge{} operation}
The edge removal do not change much in comparison to the simplified case (Listing~\ref{listing:remove-edge-complete}). The operation first reads the edge state (line~\ref{line:remove-complete-state-read}), then checks whether the edge is already absent and it can just complete (line~\ref{line:remove-complete-absent}). If the edge is spanning or is about to become spanning, it is removed under the locks (lines~\ref{line:remove-complete-blocking-start}-\ref{line:remove-complete-blocking-end}). Otherwise, the operation tries to remove an edge without taking locks (line~\ref{line:remove-complete-try}) by performing a CAS (line~\ref{line:remove-non-spanning}) and if successful, removes the edge physically (line~\ref{line:remove-complete-try-physical}). The replacement search upon spanning edge removal is discussed in a later paragraph.

\begin{figureAsListing}
  \begin{lstlisting}
func remove_edge(edge: Edge) {
  while (true):
  #\label{line:remove-complete-state-read}##\indentrule#  state := states[edge]
  #\indentrule#  // Check whether the edge is absent
  #\label{line:remove-complete-absent}##\indentrule#  if state == null or state.status == INITIAL: return
  #\indentrule#  status := state.status
  #\indentrule#  if status == SPANNING or status == IN_PROGRESS:
  #\indentrule#  #\indentrule#  blocking_remove_edge(edge), return
  #\indentrule#  if status == NON-SPANNING:
  #\label{line:remove-complete-try}##\indentrule#  #\indentrule#  if try_remove_non_spanning_edge(edge, state):
  #\indentrule#  #\indentrule#  #\indentrule#  return
}

func blocking_remove_edge(edge: Edge) {
  #\label{line:remove-complete-blocking-start}#with_components_locked(edge.u, edge.v) {
  #\indentrule#  state := states[edge]
  #\indentrule#  if state == null or state.status == INITIAL: 
  #\indentrule#  #\indentrule#  return // already absent
  #\indentrule#  if state == NON-SPANNING:
  #\indentrule#  #\indentrule#  try_remove_non_spanning_edge(edge), return
  #\indentrule#  // Iterate over levels, search for a replacement,
  #\indentrule#  // raise spanning and non-spanning edges
  #\indentrule#  remove_spanning_edge(edge)
  #\indentrule#  states.remove(edge) // physical removal
  #\label{line:remove-complete-blocking-end}#
}

func try_remove_non_spanning_edge(
  edge: Edge, state: State
): Bool { // returns whether is successful
  l := state.level 
  #\label{line:remove-non-spanning}#if states.CAS(edge, state, State(REMOVED, l)):
  #\indentrule#  // Finalize the removal.
  #\label{line:remove-complete-try-physical}##\indentrule#  F[l].remove_info(edge)
  #\indentrule#  return true
  return false
}  
  \end{lstlisting}
  \caption{Implementation of the \texttt{remove\_edge} operation. \texttt{try\_remove\_non\_spanning\_edge} tries to remove a non-spanning edge without taking locks and can fail.}
  \label{listing:remove-edge-complete}
\end{figureAsListing}

\paragraph{The \addEdge{} operation}
The basic algorithm of edge addition remains almost the same (Listing~\ref{listing:add-complete}). The only differences are the fact that we use states instead of just statuses, the states are stored in a hash map, and if an operation sees that the edge is being added by a concurrent thread as spanning, it waits for the operation completion by taking the locks (lines~\ref{line:add-complete-in-progress-read}-\ref{line:add-complete-in-progress-sync}). Besides, upon spanning edge addition the status changes twice -- in the beginning and in the end of the operation to \texttt{IN\_PROGRESS} and \texttt{SPANNING} respectively (lines~\ref{line:add-blocking-status-in-progress} and \ref{line:add-blocking-status-spanning}).

\begin{figureAsListing}
  \begin{lstlisting}
func add_edge(edge: Edge) {
  // Create a new initial state to avoid the ABA problem.
  initial_state := State(INITIAL, 0)
  // Atomically put the state only if there was none
  // or get the present state.
  prev_state := states.put_if_absent(edge, initial_state)
  if prev_state.status == INITIAL: 
  #\indentrule#  initial_state = prev_state // help to add the edge
  else if prev_state != null: return // already present
  while (true):
  #\indentrule#  state := states[edge]
  #\indentrule#  if state != initial_state:
  #\indentrule#  #\indentrule#  // Synchronize with another thread if needed
  #\label{line:add-complete-in-progress-read}##\indentrule#  #\indentrule#  if state.status == IN_PROGRESS:
  #\label{line:add-complete-in-progress-sync}##\indentrule#  #\indentrule#  #\indentrule#  with_components_locked(edge.u, edge.v) {}
  #\indentrule#  #\indentrule#  return // someone inserted the edge
  #\indentrule#  if connected(edge.u, edge.v):
  #\indentrule#  #\indentrule#  if try_add_non_spanning_edge(edge, state):
  #\indentrule#  #\indentrule#  #\indentrule#  return
  #\indentrule#  else: 
  #\indentrule#  #\indentrule#  blocking_add_edge(edge)
  #\indentrule#  #\indentrule#  return 
}  

func blocking_add_edge(edge: Edge, initial_state: State) {
  with_components_locked(edge.u, edge.v) {
  #\indentrule#  // Check whether the edge is already present
  #\indentrule#  if states[edge] != initial_state: return 
  #\indentrule#  if F[0].find_root(edge.u) != F[0].find_root(edge.v):
  #\indentrule#  #\indentrule#  // is a spanning edge
  #\label{line:add-blocking-status-in-progress}##\indentrule#  #\indentrule#  states[edge] = State(IN_PROGRESS, 0)
  #\indentrule#  #\indentrule#  F[0].add_edge(edge)
  #\label{line:add-blocking-status-spanning}##\indentrule#  #\indentrule#  states[edge] = State(SPANNING, 0)
  #\indentrule#  else:
  #\indentrule#  #\indentrule#  // is a non-spanning edge.
  #\indentrule#  #\indentrule#  // Simplified code of the non-blocking addition:
  #\indentrule#  #\indentrule#  // Publish the edge before its addition
  #\indentrule#  #\indentrule#  F[0].add_info(edge) 
  #\indentrule#  #\indentrule#  next_state = State(NON-SPANNING, 0)
  #\indentrule#  #\indentrule#  #\label{line:add-non-spanning-lp-2}#if !states.CAS(edge, initial_state, next_state):
  #\indentrule#  #\indentrule#  #\indentrule#  // already was non-spanning 
  #\indentrule#  #\indentrule#  #\indentrule#  F[0].remove_info(edge) 
  }
}
  \end{lstlisting}
  \caption{Basic parts of the \texttt{add\_edge} operation. \texttt{try\_add\_non\_spanning\_edge} tries to add a non-spanning edge without taking locks and can fail.}
  \label{listing:add-complete}
\end{figureAsListing}

The non-blocking addition (Listing~\ref{listing:nonblocking-addition}) first optimistically adds information about the edge to corresponding multi-sets (line~\ref{line:add-complete-non-bl-optimistical}), then finds the root of the Cartesian tree (line~\ref{line:add-find-root}), reads the \texttt{removal\_op} field from it that could have been published by a concurrent spanning edge removal in this component of connectivity. The ordering is important -- similarly to Theorem~\ref{theorem:non-spanning-edge-correctness}, if the information is inserted before the removal, it will see the edge, so in the only interesting case the removal holds the locks while the addition is in lines~\ref{line:add-find-root}-\ref{line:add-complete-remove-parallel-end} and the found root is correct. If the edge can be a replacement, it is proposed to a concurrent removal and the operation completes (lines~\ref{line:add-complete-propose-start}-\ref{line:add-complete-propose-end}). Otherwise, if the edge endpoints are still in the same component of connectivity, the operation is finalized by a status CAS (line~\ref{line:add-non-spanning-lp-1}).

An interesting change occurred in the replacement proposal. Previously, in the simplified case it was just a CAS to the replacement slot, but in the general case the logic is more difficult (lines~\ref{line:propose-replacement-complete-start}-\ref{line:propose-replacement-complete-end}). It reads the replacement slot (line~\ref{line:propose-replacement-read}) and then tries to update it in a lock-free style. If the removal is already over, nothing can be done (line~\ref{line:propose-replacement-CLOSED}). When the slot is empty, the operation just tries to write the edge via \texttt{CAS} (line~\ref{line:propose-replacement-cas-null}). However, if there is another edge in the slot, the method can not just complete, since the replacement can be removed by a parallel thread before its status is changed to \texttt{SPANNING} in line~\ref{line:add-complete-propose-success-cas}. So, in order not to wait for the status update, the thread can help by performing the same CAS (line~\ref{line:add-spanning-lp-2}). Because of such helping to avoid the ABA problem, the replacement slot contains not only the slot, but also its initial state. After that, in case of helping failure, the replacement slot is just cleared (line~\ref{line:propose-replacement-slot-clear}). This situation may seem bug-prone, because the removed edge can be non-spanning and still appropriate for the replacement. However, the fact that some thread managed to add the edge as non-spanning means that it started the addition and added the edge information before the removal, and thus, the removal will be able to find the information about this edge.

\begin{figureAsListing}
  \begin{lstlisting}
func try_add_non_spanning_edge(
  edge: Edge, state: State
): Bool { // returns whether is successful
  // Publish the edge info before its addition.
  #\label{line:add-complete-non-bl-optimistical}#F[0].add_info(edge)
  #\label{line:add-find-root}#root := F[0].find_root(u)
  remove_op := root.remove_op
  // Check whether there is a concurent removal.
  if remove_op != null:
  #\indentrule#  #\label{line:add-can-be-repl}#if remove_op.can_be_replacement(edge):
  #\label{line:add-complete-propose-start}##\indentrule#  #\indentrule#  if propose_replacement(remove_op, edge, state):
  #\indentrule#  #\indentrule#  #\indentrule#  F[0].remove_info(edge)
  #\indentrule#  #\indentrule#  #\indentrule#  #\label{line:add-complete-propose-success-cas}##\label{line:add-spanning-lp-1}#states.CAS(edge, state, State(SPANNING, 0))
  #\label{line:add-complete-propose-end}##\label{line:add-complete-remove-parallel-end}##\indentrule#  #\indentrule#  #\indentrule#  return true
  #\indentrule#  #\indentrule#  else if removal_operation.replacement == CLOSED:
  #\indentrule#  #\indentrule#  #\indentrule#  F[0].remove_info(edge)
  #\indentrule#  #\indentrule#  #\indentrule#  // The component is about to split, 
  #\indentrule#  #\indentrule#  #\indentrule#  // so the edge will become spanning
  #\indentrule#  #\indentrule#  #\indentrule#  blocking_add_edge(edge)
  #\indentrule#  #\indentrule#  #\indentrule#  return true #\label{line:add-replacement-end}#
  // Should the edge still be non-spanning?
  if F[0].connected(edge.u, edge.v) and 
     #\label{line:add-non-spanning-lp-1}#states.CAS(edge, state, State(NON-SPANNING, 0)):
  #\indentrule#  #\indentrule#   return true
  remove_edge_info(edge)
  return false
}

#\label{line:propose-replacement-complete-start}#func propose_replacement(
  remove_op: SpanningEdgeRemoveOp, 
  edge: Edge, state: State
): Bool {
  while (true):
  #\label{line:propose-replacement-read}##\indentrule#  current_replacement := removal_op.replacement
  #\indentrule#  // Is the replacement search already completed?
  #\label{line:propose-replacement-CLOSED}##\indentrule#  if current_replacement == CLOSED: return false
  #\indentrule#  if current_replacement == null:
  #\indentrule#  // Try to put the edge in the slot.
  #\label{line:propose-replacement-cas-null}##\indentrule#  #\indentrule#  if removal_op.replacement.CAS(null, (edge, state)):
  #\indentrule#  #\indentrule#  #\indentrule#  return true
  #\indentrule#  #\indentrule#  continue // try again
  #\indentrule#  if current_replacement.edge == edge:
  #\indentrule#  #\indentrule#  return true // already is the replacement
  #\indentrule#  (r_edge, r_state) := current_replacement
  #\indentrule#  #\label{line:add-spanning-lp-2}#if states.CAS(r_edge, r_state, State(SPANNING, 0))
       or states[r_edge].status == SPANNING:
  #\indentrule#  #\indentrule#     return false
  #\indentrule#  // Replacement edge was removed or is NON-SPANNING,
  #\indentrule#  // so can clear the slot.
  #\label{line:propose-replacement-slot-clear}##\indentrule#  removal_op.replacement.CAS(current_replacement, null)
#\label{line:propose-replacement-complete-end}#}
  \end{lstlisting}
  \caption{Non-blocking addition of a non-spanning edge.}
  \label{listing:nonblocking-addition}
\end{figureAsListing}

\paragraph{Replacement Search}
Given all the ideas we presented before, the replacement search is rather straightforward. It iterates over all current-level non-spanning edges using \texttt{has\_non\_spanning\_edges} flags and the concurrent multi-sets. For each edge it checks the status and the level of the edge. Upon finding an edge in the \texttt{INITIAL} status, i.e. an edge that is being added by a parallel thread as non-spanning it helps to add it (lines~\ref{line:replacement-search-helping-start}-\ref{line:replacement-search-helping-end}). After that, the search checks whether the edge can be a replacement and in this case writes it to its own replacement slot via the same \texttt{propose\_replacement} method (line~\ref{line:replacement-search-propose-replacement}) and, in case of success, deletes the information about the edge. If the edge is not the replacement, to amortize the cost of its finding the edge is promoted to the next ETT level (lines~\ref{line:replacement-search-promotion-start}-\ref{line:replacement-search-promotion-end}). This promotion is optimistical -- it adds the information about the edge to the next level before the CAS of its state and then removes the edge information from either the current level or the next level depending on the CAS results. When the removal finishes its search, it ``closes'' the replacement slot via \texttt{finalize\_replacement\_search} method (line~\ref{line:replacement-search-finalize}).

\begin{figureAsListing}
  \begin{lstlisting}
// Replacement search for F_0
func find_replacement(
  node: ETTNode, removal_op: RemovalOperation
): Bool {
  if node == null or !node.has_non_spanning_edges: 
  #\indentrule#  return false
  found := false
  // Iterate over all non-spanning edges.
  for edge in node.edges:
  #\indentrule#  state := states[edge]
  #\indentrule#  if state == null: continue
  #\indentrule#  if state.level != 0: continue
  #\label{line:replacement-search-helping-start}##\indentrule#  if state.status == INITIAL:
  #\indentrule#  #\indentrule#  // Help to add the edge.
  #\indentrule#  #\indentrule#  if removal_op.can_be_replacement(edge):
  #\indentrule#  #\indentrule#  #\indentrule#  // Add as a replacement.
  #\indentrule#  #\indentrule#  #\indentrule#  if propose_replacement(removal_op, edge, state):
  #\indentrule#  #\indentrule#  #\indentrule#    #\label{line:add-spanning-lp-3}#if states.CAS(edge, state, State(SPANNING, 0)):
  #\indentrule#  #\indentrule#  #\indentrule#  #\indentrule#  #\indentrule#  found = true, break
  #\indentrule#  #\indentrule#  // Should the edge still be non-spanning?
  #\indentrule#  #\indentrule#  if F[0].find_root(edge.u) == F[0].find_root(edge.v):
  #\indentrule#  #\indentrule#  #\indentrule#  F[0].add_info(edge)
  #\indentrule#  #\indentrule#  #\indentrule#  next_state: = State(NON-SPANNING, 0)
  #\indentrule#  #\indentrule#  #\indentrule#  #\label{line:add-non-spanning-lp-3}#if states.CAS(edge, state, next_state):
  #\indentrule#  #\indentrule#  #\indentrule#  #\indentrule#  state = next_state // successfully added
  #\indentrule#  #\indentrule#  #\indentrule#  else:
  #\indentrule#  #\indentrule#  #\indentrule#  #\indentrule#  #\label{line:replacement-search-helping-end}#F[0].remove_info(edge) // failure
  #\indentrule#  if state.status != NON-SPANNING: continue
  #\indentrule#  if removal_op.can_be_replacement(edge):
  #\indentrule#  #\indentrule#  if states.CAS(edge, state, State(SPANNING, 0)):
  #\label{line:replacement-search-propose-replacement}##\indentrule#  #\indentrule#  #\indentrule#  if propose_replacement(removal_op, edge):
  #\indentrule#  #\indentrule#  #\indentrule#  #\indentrule#  F[0].remove_info(edge) // replacement proposed
  #\indentrule#  #\indentrule#  #\indentrule#  else:
  #\indentrule#  #\indentrule#  #\indentrule#  #\indentrule#  states[edge] = state // revert the state
  #\indentrule#  #\indentrule#  #\indentrule#  found = true, break
  #\indentrule#  else:
  #\label{line:replacement-search-promotion-start}##\indentrule#  #\indentrule#  // Promote the non-spanning edge to the next level
  #\indentrule#  #\indentrule#  F[1].add_info(edge)
  #\indentrule#  #\indentrule#  next_state := State(NON-SPANNING, 1)
  #\indentrule#  #\indentrule#  if states.CAS(edge, state, next_state):
  #\indentrule#  #\indentrule#  #\indentrule#  F[0].remove_info(edge) // promotion successful
  #\indentrule#  #\indentrule#  else
  #\indentrule#  #\indentrule#  #\indentrule#  #\label{line:replacement-search-promotion-end}#F[1].remove_info(edge) // promotion failed
  if !found:
  #\indentrule#  found = find_replacement(node.left, removal_op)
  if !found:
  #\indentrule#  found = find_replacement(node.right, removal_op)
  recalculate_flags(node)
  return found
}

// Closes the slot and returns a replacement or null
#\label{line:replacement-search-finalize}#func finalize_replacement_search(
  removal_op: RemovalOperation
): Edge {
  if propose_replacement(removal_op, CLOSED):
    return null // the slot is closed
  return removal_op.replacement.edge
}
  \end{lstlisting}
  \caption{Replacement search for $F_0$.}
  \label{listing:replacement-search}
\end{figureAsListing}

\paragraph{Linearization Points}
The correctness of our algorithm was discussed in Theorem~\ref{theorem:non-spanning-edge-correctness} in the simplified case and the arguments remain the same for the general case. We verified our algorithms with stress testing and bounded model checking. Since our algorithm is complex, we would also like to list all possible linearization points for every \emph{successful} operation. For blocking edge modifications, we consider only linearization points of connected component changes.
\begin{itemize}
    \item \texttt{add\_edge(u, v)}
    \begin{itemize}
        \item spanning edge
        \begin{itemize}
            \item The same linearization point as for the ETT addition (Figure~\ref{linearizable-merging}).
            \item The successful \texttt{CAS} in line~\ref{line:add-spanning-lp-1} of Listing~\ref{listing:nonblocking-addition}.
            \item The successful \texttt{CAS} by a concurrent thread in line~\ref{line:add-spanning-lp-2} of Listing~\ref{listing:nonblocking-addition}.
            \item The successful \texttt{CAS} by a concurrent thread in line~\ref{line:add-spanning-lp-3} of Listing~\ref{listing:replacement-search}.
        \end{itemize}
        \item non-spanning edge
        \begin{itemize}
            \item The successful state \texttt{CAS} in line~\ref{line:add-non-spanning-lp-1} of Listing~\ref{listing:nonblocking-addition}.
            \item The successful state \texttt{CAS} in line~\ref{line:add-non-spanning-lp-2} of Listing~\ref{listing:add-complete}.
            \item The successful state \texttt{CAS} by a concurrent thread in line~\ref{line:add-non-spanning-lp-3} of Listing~\ref{listing:replacement-search}.
        \end{itemize}
    \end{itemize}
    \item \texttt{remove\_edge(u, v)}
    \begin{itemize}
        \item spanning edge
        \begin{itemize}
            \item If there is no replacement in $F_0$, then the linearization point is the same as for the ETT removal (Figure~\ref{linearizable-splitting}), otherwise components of connectivity do not change.
        \end{itemize}
        \item non-spanning edge
        \begin{itemize}
            \item The successful state \texttt{CAS} in line~\ref{line:remove-non-spanning} of Listing~\ref{listing:remove-edge-complete}.
        \end{itemize}
    \end{itemize}
    \item \texttt{connected(u, v)}
    \begin{itemize}
        \item See Theorem~\ref{reads-linearizable}, where we proved the existence of a linearization point.
    \end{itemize}
\end{itemize}

\end{document}